\documentclass[twoside]{article}

\usepackage[accepted]{aistats2026}
\usepackage{custom}
\usepackage{enumitem}

\begin{document}

\twocolumn[

\aistatstitle{Learning Linear Regression with Low-Rank Tasks In-Context}

\aistatsauthor{ Kaito Takanami \And Takashi Takahashi \And  Yoshiyuki Kabashima }

\aistatsaddress{The University of Tokyo \And The University of Tokyo\\ RIKEN AIP \And  The University of Tokyo } ]

\begin{abstract}
  In-context learning (ICL) is a key building block of modern large language models, yet its theoretical mechanisms remain poorly understood.
  It is particularly mysterious how ICL operates in real-world applications where tasks have a common structure.
  In this work, we address this problem by analyzing a linear attention model trained on low-rank regression tasks. Within this setting, we precisely characterize the distribution of predictions and the generalization error in the high-dimensional limit.
  Moreover, we find that statistical fluctuations in finite pre-training data induce an implicit regularization. 
  Finally, we identify a sharp phase transition of the generalization error governed by task structure.
  These results provide a framework for understanding how transformers learn to learn the task structure.
\end{abstract}

\section{Introduction}
\label{sec:introduction}

A defining characteristic of modern large-scale language models (LLMs) is their capacity for in-context learning (ICL). 
This allows them to perform new tasks by conditioning on a few examples provided in context, without requiring any parameter updates. 
This capability has led to impressive performance on a wide array of tasks, including few-shot classification~\cite{brown2020language}, complex reasoning~\cite{ahn2024large} and code generation~\cite{wang2023review, zhong2024can}. 
The surprising effectiveness of ICL has naturally attracted significant theoretical interest in uncovering underlying principles.

However, relatively few studies of ICL go beyond settings with independent tasks to consider structured families of tasks~\cite{oko2024pretrained}.
In practice, tasks are often not independent but share underlying structures, as widely recognized in multi-task learning (MTL)~\cite{caruana1997multitask}.  
In this setting, assuming a shared low-rank structure among tasks~\cite{han2016multi} has led to effective algorithms across domains such as computer vision~\cite{Zhang2022-mx}, natural language processing~\cite{Zhang2005-sp, Ando2005-gj}, and bioinformatics~\cite{Pong2010-zv}.  
This suggests that it is equally important to examine how ICL behaves when tasks share such common structure.

Another limitation of current ICL theory is its lack of precise analysis and interpretability.
On the one hand, existing studies can predict overall performance and reproduce key phenomena such as double descent~\cite{lu2025asymptotic} and emergence~\cite{wei2022emergent, raventos2023pretraining} in toy models, but these studies offer limited interpretability.
On the other hand, prior works have shown that transformers \textit{can} implement useful algorithms for regression~\cite{garg2022can, akyurek2022learning, von2023transformers}, yet the precise mathematical form of the predictors they \textit{actually} acquire through training has remained elusive.

In this work, we aim to gain insight into these issues by analyzing an attention-only transformer trained on linear regression tasks~\cite{garg2022can} with a low-rank structure.
Although the model is simple, it exhibits rich and nontrivial phenomena. It has therefore been extensively studied and shown to be instrumental in explaining various aspects of ICL~\cite{von2023transformers, akyurek2022learning, ahn2023transformers, Wu2023-ul, Zhang2023-nh, cui2024superiority, lu2025asymptotic,zhang2025training, Samet2025-qb} both theoretically and empirically. 
To interpret the mechanisms of ICL effectively, we situate our analysis in the high-dimensional asymptotic regime, where the system's complex behavior can be precisely characterized by a small number of macroscopic order parameters~\cite{zavatone2025summary}.

At the core of our analysis is a precise solution of a linear attention model trained on regression tasks in the high-dimensional limit.
This framework reveals a decomposition of the ICL prediction into two parts: 
an algorithmic component that performs linear regression, 
and a noise component whose effect depends on the information acquired during pre-training~(Section~\ref{sec:decomposing_the_icl_prediction}).
When the tasks share a low-rank structure, additional phenomena emerge. 
First, the model learns an efficient algorithm that exploits the low-rank structure of the tasks.  
Moreover, when the prompt contains memorized tasks or the learned task structure, the noise component can be suppressed~(Section~\ref{sec:decomposing_the_icl_prediction}).  
Second, the learning of low-rank tasks is stabilized by an implicit regularization induced by finite-data statistics during pre-training~(Section~\ref{sec:implicit_regularization}). 
Finally, we show that the rank of the pre-training tasks induces a sharp phase transition in the model’s capabilities, from a regime where the model cannot fully exploit the low-rank structure to one where the low-rank structure is fully exploited.
In the latter regime, the model faces a fundamental trade-off between specialization to pre-trained (in-distribution) tasks and robustness to unseen out-of-distribution tasks~(Section~\ref{sec:emergence_and_trade_offs_from_varying_task_difficulty})\footnote{The code is available in \url{https://github.com/taka255/icl-replica-analysis}.}.

\section{Related Work}

\paragraph{Theory of ICL.}
A growing body of theoretical studies has investigated ICL through simplified regression settings. 
Early analyses~\cite{garg2022can,akyurek2022learning,von2023transformers, ahn2023transformers} showed that Transformers can implement regression procedures in-context and characterized the learned predictors in tractable toy models. 
Subsequent works~\cite{lu2025asymptotic, letey2026pretraintest} employing asymptotic analyses of linear attention have established a high-dimensional theory that predicts macroscopic behavior. 
This perspective makes it possible to analyze the typical outcome of pretraining in tractable Transformer models, going beyond representability to characterize its generalization. 
Building on these foundations, our work focuses on structured low-rank task families and provides a mechanism-level characterization of how pretraining statistics shape the learned in-context predictor through a decomposition into an algorithmic signal and noise terms.

\paragraph{Distribution Shift Across Tasks in ICL.}
Recent work has begun to characterize ICL under distribution shift. 
\cite{kwon2026outofdistribution} show that when task vectors lie in a union of low-dimensional subspaces, ICL can generalize to any subspace in their span. 
\cite{letey2026pretraintest} study pretrain-test mismatch and show that performance degrades as the two task distributions separate, while perfect matching is not always optimal. Complementing these results, our theory links out-of-distribution degradation to an explicit structure-dependent error term that grows under task-structure mismatch.

\section{Model}
\label{sec:model}

We model the mechanism of ICL using a transformer-based architecture. 
The overall goal is to train a model that can infer and execute the algorithm for linear regression purely from a sequence of examples provided in its context. 
To this end, we formalize a problem setup consisting of a pre-training phase, where the model learns the general algorithm, 
and an inference phase, where its ability to generalize to new instances is evaluated\footnote{A comprehensive list of notations is provided in Appendix~\ref{appendix:notations} for reference.}.

\subsection{Pretraining Phase} \label{sec:pretraining_phase}

We begin by constructing a base collection of tasks $\mathcal{W}_0 = \{\vb{w}^\mu \in \mathbb{R}^D \mid \mu=1, \ldots, M_0\}$, which represents a diverse pool of possible task vectors. Each $\vb{w}^\mu$ is generated from a low-dimensional latent structure: we fix a shared feature matrix $A \in \mathbb{R}^{D \times r}$ with orthonormal columns ($r \le D$), draw a latent task feature vector $\vb{v}^\mu \sim \mathcal{N}(\mathbf{0}, I_r)$, and map it to the $D$-dimensional space via $\vb{w}^\mu = \sqrt{D/r} A \vb{v}^\mu$.

From this base collection, we then form the actual training set $\mathcal{W}$ by sampling $M (\gg M_0)$ tasks uniformly with replacement. This step reflects the fact that, in the thermodynamic limit, obtaining nontrivial learning behavior requires that individual tasks may recur many times. In other words, $\mathcal{W}_0$ captures the diversity of available tasks, while $\mathcal{W}$ specifies the effective workload that the learner repeatedly encounters.

For each task $\vb{w}^\mu \in \mathcal{W}$, we create a training instance
by first drawing $L+1$ input vectors, $\{\vb{x}_i^\mu\}_{i=1}^{L+1}$,
independently from a normal distribution $\mathcal{N}(\mathbf{0}, \mathbf{I}_D/D)$.
The corresponding outputs are then generated as $y_i^\mu = \vb{w}^\mu \cdot \vb{x}_i^\mu + \epsilon_i^\mu$,
where each $\epsilon_i^\mu$ is an independent noise term drawn from $\mathcal{N}(0, \sigma^2)$.
The first $L$ pairs, $\{(\vb{x}_l^\mu, y_l^\mu)\}_{l=1}^L$, serve as \textbf{demonstrations},
while the final pair, $(\vb{x}_{L+1}^\mu, y_{L+1}^\mu)$, constitutes the \textbf{query} and its \textbf{label}.
The entire generative process for these instances is denoted by $\mathcal{D}_{\text{pretrain}}$.
These elements are concatenated to form the input sequence for the model, which we refer to as the \textbf{context} $C^\mu \in \R^{(D+1) \times (L+1)}$:
\begin{equation}
  C^\mu = \ab( \mqty{\vb{x}_1^\mu & \vb{x}_2^\mu & \cdots & \vb{x}_L^\mu & \vb{x}_{L+1}^\mu \\
  y_1^\mu & y_2^\mu & \cdots & y_L^\mu & 0} ) ,
\end{equation}
where the label corresponding to the query $\vb{x}_{L+1}^\mu$ is replaced with a placeholder value of zero, indicating that it is the unknown quantity the model is required to predict.

Our model is an attention-based architecture that maps an input context to a prediction. This is expressed as the composition of an attention function $\mathsf{Atten}_{\Theta}: \R^{(D+1) \times (L+1)} \to \R^{(D+1) \times (L+1)}$ with parameters $\Theta$, and a readout function $\mathsf{Read}: \R^{(D+1) \times (L+1)} \to \R$. The prediction for the query is given by:
\begin{equation}
  \hat{y}_{L+1}^\mu = \mathsf{Read}(\mathsf{Atten}_{\Theta}(C^\mu)) .
\end{equation}
The model is trained by minimizing the mean squared error (MSE) loss between the predictions and the true labels over all $M$ training instances. The loss function is:
\begin{equation}
  \mathcal{L}(\Theta) = \frac{1}{M} \sum_{\mu=1}^M ( \hat{y}_{L+1}^\mu - y_{L+1}^\mu )^2 . \label{eq:loss_function}
\end{equation}
Through optimization, we obtain the learned parameters $\Theta^* = \arg\min_{\Theta} \mathcal{L}(\Theta)$.

\subsection{Inference Phase}
\label{sec:inference_phase}

During the inference phase, we evaluate the generalization ability of the pre-trained model
with fixed optimal parameters $\Theta^*$.
Performance is assessed on new test instances generated from a fixed evaluation task $\vb{w}^* \in \mathbb{R}^D$.
Each test instance is created by first drawing $\tilde{L}+1$ input vectors, $\{\vb{x}_i\}_{i=1}^{\tilde{L}+1}$,
independently from $\mathcal{N}(\mathbf{0}, \mathbf{I}_D/D)$,
where $\tilde{L} \le L$ reflects a few-shot learning scenario.
The corresponding outputs are generated without noise as $y_i = \vb{w}^* \cdot \vb{x}_i$
to assess the model's pure algorithmic reasoning.
The first $\tilde{L}$ pairs serve as the \textbf{demonstrations},
while the final pair serves as the \textbf{query} $\vb{x}^{\tilde{L}+1}
=\vb{x}$ and its \textbf{label} $y^{\tilde{L}+1}=y$.
We refer to this test data distribution for a given task $\vb{w}^*$ as $\mathcal{D}_{\text{test}}$.
Since the length of demonstrations $\tilde{L}$ may be less than $L$, the input sequence, referred to as the \textbf{prompt} $P$, is padded with zero vectors to match the context length of $L+1$ expected by the model. 
The prompt $P \in \R^{(D+1) \times (L+1)}$ is constructed as:
\begin{equation}
P = \ab( \mqty{\vb{x}_1 & \cdots & \vb{x}_{\tilde{L}} & \vb{0} & \cdots & \vb{0} & \vb{x} \\
y_1 & \cdots & y_{\tilde{L}}  & 0 & \cdots & 0 & 0} ) ,
\end{equation}
where the final column is the query input $\vb{x}$ with its label hidden. 
The prediction $\hat{y}$ is compared with the true label $y$.
The model's prediction for the query is obtained by applying the frozen network; $\hat{y} = \mathsf{Read}(\mathsf{Atten}_{\Theta^*}(P))$.

%%%%%%%%%%%%%%%%%%%%%%%%%%%%%%%%%%%%%%%%%%%%%%%%%%%%%%%%%%%%

\subsection{Evaluation Protocol}
\label{sec:evaluation_protocol}

We evaluate the model's ICL capabilities across three distinct protocols.
\begin{itemize}[leftmargin=*, nosep]
    \item \textbf{Task Memorization (TM):} The task is sampled from the pre-training set, $\vb{w}^* \sim \mathsf{Unif}(\mathcal{W}_0)$. This protocol tests the model's ability to recall and execute a task seen during training, given a new set of demonstrations.
    \item \textbf{In-Distribution Generalization (IDG):} The task is novel but generated from the same underlying structure, $\vb{w}^* = \sqrt{D/r} A \vb{v}^*$, for a new latent vector $\vb{v}^* \sim \mathcal{N}(\mathbf{0}, I_r)$. This protocol measures the ability to generalize to unseen tasks that share the same low-dimensional structure as the training tasks.
    \item \textbf{Out-of-Distribution Generalization (ODG):} The task is drawn from a standard isotropic Gaussian distribution, $\vb{w}^* \sim \mathcal{N}(\vb{0}, I_D)$. These tasks, by design, lack the low-rank structure inherent in the training distribution, thereby testing the model's robustness to a fundamental structural mismatch.
\end{itemize}
The model's performance on each protocol is quantified by its specific MSE, which we denote generically as $\mathcal{E}_{\text{protocol}}$ for $\text{protocol} \in \{\text{TM, IDG, ODG}\}$. 
This error is formally defined as:
\begin{equation}
  \mathcal{E}_{\text{protocol}} = \E_{\vb{w}^* \sim \mathcal{D}_{\text{protocol}}} \E_{\mathcal{D}_{\text{test}}, \mathcal{D}_{\text{pretrain}}}   \ab[ (y - \hat{y})^2],
\end{equation}
where $\mathcal{D}_{\text{protocol}}$ represents the task distribution corresponding to each protocol as defined above.

\subsection{Attention Architecture}

For analytical tractability, we now specify the general architecture to a single-layer linear attention model. 
In this simplified setting, the attention mechanism is defined as:
\begin{equation}
\mathsf{Atten}_{\Theta} (C) = C + \frac{1}{L} VC (KC)^\top (QC), \label{eq:attention_model}
\end{equation}
where the learnable parameters are $\Theta = \{V, K, Q\}$, which are matrices in $\mathbb{R}^{(D+1) \times (D+1)}$. The readout function simply extracts the final scalar output at the query position: $\mathsf{Read}(C') = C'_{D+1, L+1}$.

Under this setting and certain simplifying assumptions, it is known that the pre-training process is equivalent to training a one-layer neural network (see  \cite{Zhang2023-nh, zhang2025training} and Appendix~\ref{app:eq-lenear-nn}). 
Specifically, the input context $C^\mu$ can be mapped to an effective feature matrix $H^\mu \in \R^{D \times D}$ for each training instance $\mu$:
\begin{equation}
  H^\mu = \frac{D}{L} \ab[\vb{x}^{\mu}_{L+1} \sum_{l \le L} y_{l}^\mu (\vb{x}^{\mu}_l)^\top] . \label{eq:H_def}
\end{equation}
The attention model's prediction is then equivalent to the output of a neural network, which is a linear function of this effective feature:
\begin{equation}
\hat{y}_{L+1}^\mu = \mathrm{Tr}\ab(W  {H^\mu}^\top) = \sum_{i,j=1}^D W_{ij} H_{ij}^\mu, \label{eq:equivalent_NN}
\end{equation}
where the trainable weights of the attention mechanism $\{V, K, Q\}$ are implicitly mapped to a single effective weight matrix $W \in \R^{D \times D}$.

This equivalence extends directly to the inference phase. For a given prompt $P$ with $\tilde{L}$ demonstrations, the effective feature matrix $H$ and the corresponding prediction $\hat{y}$ are given by:
\begin{equation}
  H = \frac{D}{\tilde{L}} \ab[\vb{x} \sum_{l \le \tilde{L}} y_{l} (\vb{x}_l)^\top] \quad \hat{y} = \mathrm{Tr}\ab( W^* H^\top ), \label{eq:H_def_inference}
\end{equation}
where $W^*$ is the learned weight matrix from the pre-training phase.

\subsection{High-Dimensional Analysis}

To analyze the model's typical performance, we consider the high-dimensional limit where the problem dimensions ($D, L, \tilde{L}, M_0, M, r$) grow to infinity. The system's behavior in this regime is characterized by several dimensionless parameters, which are assumed to be fixed constants in $(0, \infty)$. We define the \textbf{Pre-training Sample Ratio} as $\alpha = L/D$ and the \textbf{Inference Sample Ratio} as $\tilde{\alpha} = \tilde{L}/D$, which represent the ratio of examples to features during pre-training and at inference, respectively, with the constraint that $\tilde{\alpha} \leq \alpha$. Furthermore, we introduce the \textbf{Task Difficulty} $\rho = r/D$, which is the relative dimension of the latent task space where $\rho \in (0, 1]$. Finally, we define the \textbf{Task Diversity} as $\kappa = M_0 / D$ to capture the richness of the base task set, and the \textbf{Training Data Density} as $\gamma = M / (D M_0)$ to represent the number of training instances per task.

To characterize the model's behavior, we define two regimes based on data density and task diversity. 
For data density, we distinguish the data-rich regime ($\kappa \gamma > 1$) and the data-deficient regime ($\kappa \gamma < 1$). 
Independently, for task diversity, we distinguish the \textit{task-rich regime} ($\rho < \kappa$), where the variety of tasks spans the latent space, and the task-deficient regime ($\rho > \kappa$), where it does not.

\section{Decomposition of ICL Prediction}
\label{sec:decomposing_the_icl_prediction}
Our high-dimensional analysis reveals how structured tasks are learned in context.
The model's prediction decomposes into a core algorithmic signal and two noise terms that can be suppressed by context.
This view reinterprets ICL as a context-dependent noise-reduction mechanism, providing a framework for how general-purpose models adapt their computation.
Formally, the prediction is given by the following result:
\begin{result}[Decomposition of ICL prediction]
\label{result:decomposition}
  The prediction $\hat{y}$ for a given prompt is:
  \begin{equation}\label{eq:decomp}
      \hat{y} = \operatorname{tr}(W^* H^\top) \overset{\mathrm{d}}{=} \hat{y}_{\mathrm{algo}} + \hat{y}_{\mathrm{mem}} + \hat{y}_{\mathrm{struct}},
  \end{equation}
  where $\overset{\mathrm{d}}{=}$ represents equality in distribution.
  This decomposition yields a deterministic \textbf{algorithmic signal} ($\hat{y}_{\mathrm{algo}}$) and two independent noise terms: a \textbf{memorization noise} ($\hat{y}_{\mathrm{mem}}$) correlated with the pre-training tasks, and a \textbf{structure noise} ($\hat{y}_{\mathrm{struct}}$) uncorrelated with them. This decomposition holds for any pre-training task distribution $\mathcal{W}_0$\footnote{We use a replica method~\cite{mezard1987spin, charbonneau2023spin} to analyze the asymptotic limit. This method involves a non-rigorous mathematical step. Consequently, we conservatively state our findings as ``Results'' rather than ``Theorems''. However, the validity of our theoretical predictions is confirmed by numerical experiments, which show excellent agreement (see Appendix~\ref{app:agreement}).}.
\end{result}
The next result characterizes each component analytically for $\alpha \gg 1$.
\begin{result}[Asymptotic Characterization of the Components]
\label{result:asump_decomp}
  Assume data-rich and task-rich regimes. Each term in Eq.~\eqref{eq:decomp} is given by:
  {\small
  \begin{align}
    \hat{y}_{\mathrm{algo}} &\overset{\mathrm{d}}{=} \ab(P_A \vb{w}_{\mathrm{prompt}})^\top \vb{x} + \mathcal{O}\ab(\frac{1}{\alpha}) \label{eq:algo_term}\\
    \hat{y}_{\mathrm{mem}} &\overset{\mathrm{d}}{=} \frac{\sigma\sqrt{\rho}}{\gamma-1/\kappa} Z_{\mathrm{mem}} + \mathcal{O}\ab(\frac{1}{\alpha}) \label{eq:mem_term}\\
    \hat{y}_{\mathrm{struct}} &\overset{\mathrm{d}}{=} \frac{\sigma\rho}{\sqrt{\ab(1+\sigma^2)\ab(\gamma-1/\kappa)}} Z_{\mathrm{struct}} \sqrt{\alpha} + \mathcal{O}\ab(\frac{1}{\sqrt{\alpha}}) \label{eq:gen_term},
  \end{align}
  }
  where random variables $Z_{\mathrm{mem}}$ and $Z_{\mathrm{struct}}$ are given by:
  {\small
  \begin{align}
    Z_{\mathrm{mem}} &\sim \mathcal{N}\ab(0, \ab(1/M_0) \ab(A^\top \vb{w}_{\mathrm{prompt}})^\top S_v^{-1} \ab(A^\top \vb{w}_{\mathrm{prompt}})) \label{eq:Z_mem_def}\\
    Z_{\mathrm{struct}} &\sim \mathcal{N}\ab(0, \ab(1/M_0) \abs{P_A^{\perp} \vb{w}_{\mathrm{prompt}}}^2) \label{eq:Z_gen_def}
  \end{align}
  }
  respectively. Here, $S_v = (1/M_0) \sum_{\mu} \vb{v}^\mu {\vb{v}^\mu}^\top, P_A = AA^\top, P_A^{\perp} = I_D - P_A$ and $\vb{w}_{\mathrm{prompt}} = 1/\tilde{\alpha} \sum_l y_l \vb{x}_l$.
\end{result}
For the formal statement and derivation of Results.~\ref{result:decomposition} and \ref{result:asump_decomp}, see Appendix~\ref{appendix:replica}.

\begin{figure*}[ht] 
  \includegraphics[width=1.0\textwidth]{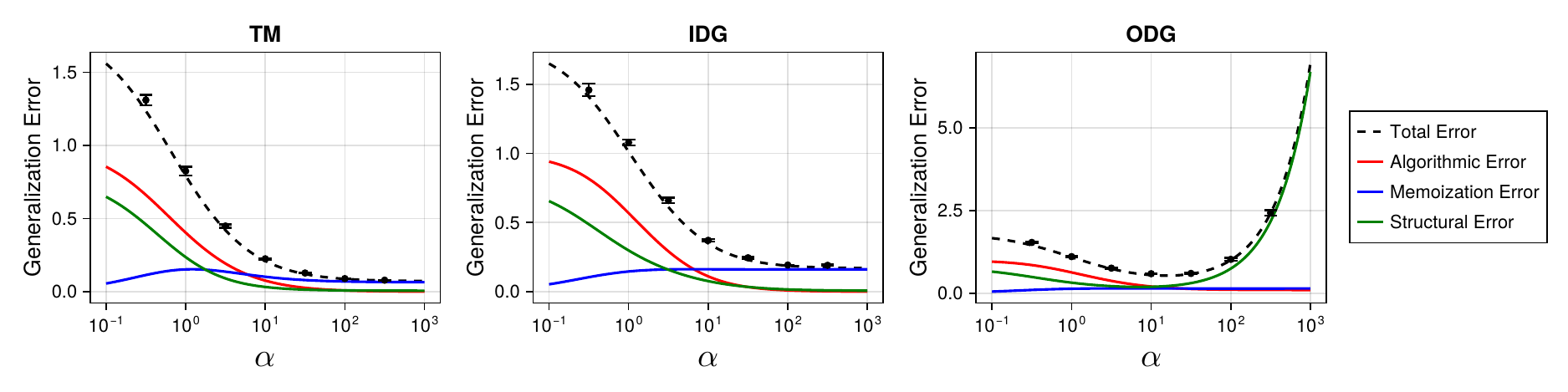}
  \caption{\textbf{Decomposition of generalization error reveals ICL's noise-reduction mechanism.} 
  The total generalization error (in theory with dashed lines, and in numerical experiments with error bars) 
  and its decomposition into algorithmic (signal), memorization (noise), and structural (noise) components (solid lines) for each protocol.
  Parameters: $\gamma=1.5, \kappa=1.5, \rho=0.9, \sigma=0.3, \tilde{\alpha}=\alpha, D = 60$.
  The error bars represent the standard errors of the mean over $5$ independent numerical experiments per point.}
  \label{fig:decomposed_error}
\end{figure*}

To interpret Result~\ref{result:decomposition} at the level of generalization, we convert the prediction decomposition into an error decomposition.  
Because the two noise terms are independent and centered, their cross terms vanish after averaging, and we obtain $\mathcal{E}_{\text{protocol}} = \E[(y-\hat{y})^2] = \E[(y-\hat{y}_{\mathrm{algo}})^2] + \E[\hat{y}_{\mathrm{mem}}^2] + \E[\hat{y}_{\mathrm{struct}}^2].$
Figure~\ref{fig:decomposed_error} plots these three contributions, showing how the decomposition of the predictor translates directly into the generalization error.
These decompositions reveal that the prediction is governed by one algorithmic signal and two distinct noise sources, which we explain in detail below.

\paragraph{The Algorithmic Signal ($\hat{y}_{\text{algo}}$):}
The leading term, $\hat{y}_{\text{algo}}$, represents the model's core algorithmic competency. 
Eq.~\eqref{eq:algo_term} reveals that ICL can acquire a sophisticated two-step procedure learned implicitly from pre-training without any explicit regularization.
First, the model computes a naive (Matched Filter) estimator $\vb{w}_{\text{prompt}} = 1/\tilde{\alpha} \sum_l y_l \vb{x}_l$ from the prompt examples. 
Second, it projects this estimate onto the low-rank subspace defined by $P_A = AA^\top$, which represents the structural prior learned from the pre-training task distribution. 
This term is the primary signal of the prediction. 
As shown in Figure~\ref{fig:decomposed_error}, the error contribution from this algorithmic component consistently decreases as $\alpha$ increases, indicating that a long context improves the model's algorithmic accuracy.

\paragraph{The Suppressible Noise Terms ($\hat{y}_{\text{mem}}, \hat{y}_{\text{struct}}$):}
Our analysis shows that the model's versatility across tasks inherently introduces noise, and that ICL operates by selectively suppressing these noise terms through context. This dual effect explains how ICL supports general-purpose adaptability, and the following two terms illustrate how it manifests in practice.

The term $\hat{y}_{\text{mem}}$ represents a memory-dependent noise. 
Eq.~\eqref{eq:mem_term} shows that its magnitude is determined by the squared Mahalanobis distance between the projected task estimate $A^\top \vb{w}_{\text{prompt}}$ and the memorized task matrix $S_v$ in the latent space.
This distance becomes small when the current task estimate is statistically typical of the memorized tasks, i.e., it aligns with directions in which the tasks $\{\vb{v}^\mu\}$ have high variance.
As seen in Figure~\ref{fig:decomposed_error}, this memorization noise is lower for TM tasks compared to unfamiliar IDG tasks, as the model can successfully recall the familiar task. 
However, this term reveals a trade-off with respect to $\alpha$. 
While a longer context reduces algorithmic and structural errors, it paradoxically worsens this noise, as shown in Figure~\ref{fig:decomposed_error} in the IDG protocol. 
This is a consequence of task-overfitting: with large $\alpha$, the model's over-specialization to the specific patterns of pre-training instances leads to poorer generalization on new tasks, thereby increasing the memorization error.

The term $\hat{y}_{\text{struct}}$ represents a structure-dependent noise. 
The estimator $P_A^{\perp} \vb{w}_{\text{prompt}}$ in Eq.~\eqref{eq:gen_term} explicitly measures the component of the in-context estimate $\vb{w}_{\text{prompt}}$ that is orthogonal to the learned structure $A$. 
This term can be interpreted as an error signal generated by the portion of the prompt that violates the learned structural prior. 
Figure~\ref{fig:decomposed_error} clearly shows this: for ODG tasks, the fundamental structural mismatch causes this noise component to become pronounced and dominate the error. 
Conversely, for in-distribution tasks (TM and IDG), a longer context allows the model to better identify the underlying structure, leading to a decrease in this structural noise.

The decomposition connects to practical ICL settings. The memorization and structural noise terms explain empirical patterns~\cite{mueller-etal-2024-context, wang2025can, letey2026pretraintest}, including strong performance on familiar or aligned tasks and degradation on unfamiliar or mismatched ones. Furthermore, the growth of memorization noise with longer context in in-distribution settings explains the task-overfitting-driven performance drop, a phenomenon that has received limited attention in prior work. In this way, the decomposition matches existing empirical findings and generates new testable predictions.

\section{Finite Pre-training Data Induces Implicit Regularization}
\label{sec:implicit_regularization}

Section~\ref{sec:decomposing_the_icl_prediction} showed that the model learns a low-rank regression algorithm; here, we investigate how this learning remains stable.
This is particularly puzzling in low-rank settings ($\rho < 1$) where, despite the absence of explicit regularization in Eq.~\eqref{eq:loss_function}, our model achieves robust generalization.
To understand this mechanism, we first examine the \textbf{idealized limit} of infinite pre-training context length ($\alpha \to \infty$ before learning).
Contrary to the intuition that perfect data should improve performance, we demonstrate that this limit, lacking finite-sample fluctuations, renders the learning problem ill-posed.
This leads to an initialization-sensitive solution and poor generalization, a stark contrast to the stable learning in the \textbf{asymptotic limit} ($\alpha\to\infty$ after learning).
Our theoretical explanation attributes this stability to an implicit regularization induced by finite-data statistics.

\subsection{Instability in the Idealized Limit}

We conduct a numerical experiment to compare the model's IDG error under two conditions: a standard setting with a finite number of demonstrations $\alpha$, and an idealized limit. 
The latter is implemented by training on the expected feature matrix $H_{\text{ideal}}^\mu = \lim_{\alpha\to\infty} [H^\mu] = \vb{x}_{L+1}^\mu {\vb{w}^\mu}^\top$. 
In both cases, the model is trained via gradient descent from a random initialization $\vb{w}_0$ satisfying $\| \vb{w}_0 \| = D$. 
The resulting IDG error is plotted as a function of $\rho$ in Figure~\ref{fig:implicit_reg_experiment_GD}. The plot shows that for any $\rho < 1$, the error in the idealized limit fails to decrease, 
whereas for finite $\alpha$, the error systematically improves as $\alpha$ increases. 
This result is counter-intuitive; the idealized model, which is given perfect information about the ground-truth task $\vb{w}^\mu$ through $H_{\text{ideal}}^\mu$, was expected to perform better, 
yet it is outperformed by the model that must infer the task from finite demonstrations.

\begin{figure}[htb] 
  \includegraphics[width=0.5\textwidth]{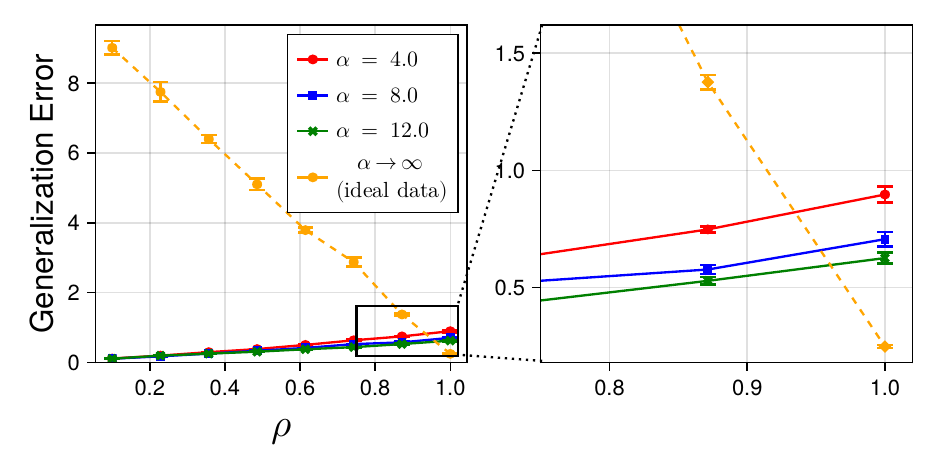}
  \caption{\textbf{Implicit regularization from finite data prevents learning instability.} 
  IDG error as a function of the task subspace dimensionality $\rho$. 
  The figure compares the performance of the idealized limit (dashed line), with standard setting using finite sample ratio $\alpha$ (solid lines). 
  %While the idealized limit fails to learn for any rank-deficient structure ($\rho < 1$), standard training is successful, with the error systematically improving as $\alpha$ increases.
  Parameters: $D = 40, M_0 = 60, \tilde{L} = 160, M = 2400, \sigma = 0.01$. 
  The error bars represent the standard errors of the mean over $5$ trials per point.}
  \label{fig:implicit_reg_experiment_GD}
\end{figure}

\subsection{Theoretical Explanation}

To explain this phenomenon, our theoretical analysis reveals that the learned weight matrix $W^*$ behaves as if it were the solution to an effective optimization problem, which differs critically depending on whether the idealized limit or the asymptotic limit is considered.

\begin{result}[Effective Objective Function] \label{result:implicit_regularization}
 There exist scalars $\hat{q}$ and $\hat{\bar{q}}$, determined solely by the system parameters, such that the statistics of the learned weight matrix $W^*$ are equivalent to those of the solution to the following optimization problem, i.e.,
   $W^* \overset{\mathrm{d}}{=} \operatorname*{argmin} f(W)$,
 where the effective objective function $f(W)$ is given by:
 \begin{equation}
   f(W) = \tr\ab[\frac{\hat{q}}{2} WSW^\top - \mathcal{M}_S W^\top + \frac{\hat{\bar{q}}}{2}WW^\top]. \label{eq:implicit_regularization_f}
 \end{equation}
 Here, the matrix $S = (1/M_0)\sum_{\mu=1}^{M_0} \vb{w}^\mu{\vb{w}^\mu}^\top$ is the task covariance matrix, and $\mathcal{M}_S$ is a random noise matrix whose structure depends on $S$. 
\end{result}

Within Eq.~\eqref{eq:implicit_regularization_f}, the two scalars $\hat{q}$ and $\hat{\bar{q}}$ arise from marginalizing over the statistical fluctuations in the multi-body problem Eq.~\eqref{eq:loss_function}. 
The parameter $\hat{q}$ represents the strength of the quadratic interaction between the model parameters and the task structure encoded by $S$. 
The parameter $\hat{\bar{q}}$ acts as an effective regularization strength. 
This regularization is essential for learning low-rank tasks ($\rho < 1$), where the task covariance matrix $S$ is rank-deficient.
Without it, the objective function possesses flat directions corresponding to the null space of $S$.
We therefore next characterize the magnitude of this effective regularization.

\begin{result}[Asymptotics of Regularization]
\label{result:asymptotic_coefficients}
 When $\kappa\gamma>1$, the scalars in the effective objective function (Result~\ref{result:implicit_regularization}) scale as $\hat{q} = \mathcal{O}(1)$ and $\|\mathcal{M}_S\| = \mathcal{O}(1)$.
 The effective regularization coefficient $\hat{\bar{q}}$ is given by:
 \begin{equation}
  \hat{\bar{q}} =
  \begin{cases}
   \mathcal{O}(1/\alpha) > 0 & (\text{in the asymptotic limit}) \\
   0 & (\text{in the idealized limit}).
  \end{cases}
 \end{equation}
\end{result}
For the formal statement and derivation, see Appendix~\ref{appendix:replica}.

These results provide a clear explanation for the observed phenomenon.
Result~\ref{result:asymptotic_coefficients} reveals that the regularization emerges only when the pre-training dataset is finite ($\hat{\bar{q}} > 0$).
By contrast, in the idealized limit, \(\hat{\bar{q}} = 0\), so the minimization problem becomes ill-posed and may fail to have a unique minimizer, leading to the learning instability seen in Figure~\ref{fig:implicit_reg_experiment_GD}.

Furthermore, this framework explains why performance improves as $\alpha$ increases.
The effective regularization strength, given by the ratio $\hat{\bar{q}}/\hat{q} \sim \mathcal{O}(1/\alpha)$, is naturally annealed as more data becomes available.
This allows the model to rely more on the data and converge to a more accurate solution while avoiding instability.

\subsection{Implications}

This finding has several significant implications for theory and practice.

\paragraph{ICL as a Two-step Automatic Learning Process.}
We hypothesize that the origin of this implicit regularization lies in the shared structure and the two-stage process inherent to ICL. 
We can think of the first stage as task identification, where the model identifies the task from the pre-training dataset. 
The statistical fluctuations in this data ensure the resulting internal task representation is never perfectly low-rank and always contains a small, full-rank noise component. 
This unavoidable noise then becomes a beneficial regularizer during the second stage, which we can call inference execution. 
It breaks the flat directions of the optimization landscape that would otherwise cause learning instability in the idealized limit.
In this way, the statistical imperfection in identifying the task is precisely what provides the stability for its execution in low-rank tasks.

\paragraph{Data as a Regularizer.}
In deep learning, implicit regularization is often attributed to optimization algorithms like SGD~\cite{zhang2017understanding}. 
While Transformers are also known to exhibit implicit regularization~\cite{Vasudeva2024-og, Frei2024-fe}, 
existing theories often treat it as an extension of standard supervised learning. 
In contrast, our findings reveal a distinct form of regularization originating from the statistical fluctuations of a finite dataset. We attribute this phenomenon to the unique ICL objective of automatically acquiring an algorithm.
This implies that for certain structures, learning is ironically hindered by idealized, noise-free data, suggesting that the natural statistical noise in a dataset can be a crucial feature, not a bug to be removed.

\section{Trade-offs from Task Difficulty}
\label{sec:emergence_and_trade_offs_from_varying_task_difficulty}
Building on our finding that the model learns a stable low-rank regression algorithm, we now explore how the model's capability is governed by the properties of this low-rank structure. 
We develop a theoretical framework based on eigenvalue spectral analysis, which predicts a transition in the model’s capabilities. This prediction is later validated by our results, revealing a sharp phase transition and fundamental trade-off between specialization and out-of-distribution robustness.

\subsection{Spectral Analysis of Task Structure}

To understand the model's behavior, we draw an analogy to classical linear regression, where the Gram matrix is $G=X^\top X$. 
In the eigenbasis of $G$, the variance of the estimated coefficient along the direction associated with eigenvalue $\lambda_i$ is $\mathsf{Var}(\hat{w}_i)\propto \sigma^2 / \lambda_i$, where $\sigma$ is the noise variance.
A large eigenvalue of $G$ corresponds to a feature direction with a high signal-to-noise ratio, leading to statistical stability.
A small eigenvalue indicates an ill-conditioned direction that is highly sensitive to noise. 
Thus, the number and magnitude of nonzero eigenvalues jointly determine the model’s learnable patterns and their stability.
A zero eigenvalue means that no learnable pattern exists in the corresponding direction, so any apparent estimate comes only from regularization or noise.

Applying this analogy to our ICL framework, we can explain the model's capabilities by analyzing the eigenvalue distribution of the task covariance matrix $S = (1/M_0) \sum_{\mu=1}^{M_0} \vb{w}^\mu{\vb{w}^\mu}^\top$ from Eq.~\eqref{eq:implicit_regularization_f}, which encapsulates the task information learned during pre-training. 
This follows from the fact that, in the eigenbasis of $S$ with eigenvalues $s_j$, the variance of the corresponding estimator scales as $\mathsf{Var}(W_{ij}) \propto 1/(\hat{q} s_j + \hat{\bar{q}})^2$ from Eq.~\eqref{eq:implicit_regularization_f}.
The following theorem precisely characterizes this distribution of eigenvalues and predicts that its structure changes at the transition point $\rho=\kappa$.

\begin{theorem}[Spectrum of the Task Matrix]\label{theorem:eigenvalue_distribution_of_S}
  Recall that $\rho=r/D$ and $\kappa=M_0/D$. The empirical spectral distribution of $S$ converges almost surely to
  \begin{align}
    \nu_S(s) &= \left(1 - \min(\rho,\kappa)\right) \delta(s) \notag \\
    & + \mathbf{1}_{[s_-,s_+]}(s) \, \frac{\rho\kappa}{2\pi s} \sqrt{(s_+ - s)(s - s_-)}.
  \end{align}
  Here, $\mathbf{1}_{[s_-,s_+]}(s)$ is the indicator function. The support edges are given by $s_\pm =  \left(\sqrt{\rho} \pm \sqrt{\kappa}\right)^2 / (\rho\kappa)$.
\end{theorem}
The proof is given in Appendix~\ref{app:eigen_proof}.

This theorem provides a quantitative basis for predicting the model's performance by revealing how the learning bottleneck shifts.
In the task-deficient regime ($\rho > \kappa$), the number of learnable patterns is limited by $\kappa$ (the fraction of nonzero eigenvalues).
Moreover, the typical eigenvalue magnitude scales as $\mathcal{O}(1/\kappa)$, implying that the statistical stability of these patterns is also controlled by $\kappa$.
As a result, the model's capabilities do not change even as tasks get simpler, indicating the model's learning is bottlenecked by task diversity.
In contrast, in the task-rich regime ($\rho < \kappa$), the number of nonzero eigenvalues is $\rho$, and their magnitudes scale as $\mathcal{O}(1/\rho)$.
As tasks become simpler, the larger eigenvalues concentrate useful information, allowing the model to fully exploit the simplified low-rank structure for more stable and efficient learning.

\subsection{Phase Transition in ICL Capabilities}
To connect our spectral analysis with practical ICL performance, we derive the theoretical generalization errors for three settings: TM, IDG, and ODG.

\begin{result}
\label{result:gen_error_simple}
    The generalization errors $\mathcal{E}_{\mathrm{TM}}$, $\mathcal{E}_{\mathrm{IDG}}$, and $\mathcal{E}_{\mathrm{ODG}}$ can be analytically obtained by solving a finite system of self-consistent equations.
    In particular, when the asymptotic case $\alpha\to \infty$, $\kappa\gamma>1$, and $\sigma=0$, each metric is given by:
    \begin{align}
        \mathcal{E}_{\mathrm{TM}} &= \frac{\min(\rho, \kappa)}{\tilde{\alpha}} \left( 1 + \frac{1-\min(\rho, \kappa)}{\kappa\gamma-1} \right) \\
        \mathcal{E}_{\mathrm{IDG}} &= \begin{cases}
            \mathcal{E}_{\mathrm{TM}} & \text{for } \rho < \kappa \\
            1 - \frac{\kappa}{\rho} - \frac{\kappa(\kappa-\rho)}{\rho(\kappa \gamma-1)} + \mathcal{E}_{\mathrm{TM}} & \text{for } \rho \ge \kappa
            \end{cases} \\
        \mathcal{E}_{\mathrm{ODG}} &= 1 - 2\min(\rho, \kappa) + (1 + \tilde{\alpha}) \mathcal{E}_{\mathrm{TM}}.
    \end{align}
\end{result}

For the formal statement and derivation, see Appendix~\ref{appendix:replica}.

As shown in Result~\ref{result:gen_error_simple}, these errors exhibit a sharp phase transition at the predicted critical point $\rho = \kappa$ when $\alpha \to \infty$.
This extends prior work that identified phase transitions with respect to task diversity when no shared task structure is assumed~\cite{lu2025asymptotic}.
Figure~\ref{fig:plateau} illustrates this behavior by showing how the generalization errors change as $\rho$ varies with fixed $\kappa$: for finite $\alpha$ we observe a crossover, which sharpens into a distinct phase transition as $\alpha \to \infty$, dividing the generalization behavior into two regimes.
In the task-deficient regime ($\rho > \kappa$), both TM and ODG errors remain on a plateau, consistent with our spectral analysis.
In this regime, the model’s performance is unaffected by task difficulty, since learning is bottlenecked by task diversity.  
In the task-rich regime ($\rho < \kappa$), as task difficulty decreases, the ODG error increases, while the IDG and TM errors decrease in tandem and remain identical.  
This outcome demonstrates that the model can perfectly capture the in-distribution task structure when $\alpha \to\infty$, as reflected in the alignment of IDG and TM errors.  
It also shows that strong specialization to the low-rank task structure enhances memorization and in-distribution performance, but at the expense of out-of-distribution robustness.  

Next, we analyze the role of $\tilde{\alpha}$.  
Result~\ref{result:gen_error_simple} shows that in the task-rich regime the $\mathcal{E}_{\mathrm{TM}}$ and $\mathcal{E}_{\mathrm{IDG}}$ scale as $\mathcal{O}(1/\tilde{\alpha})$.  
This means the model has learned the in-distribution task structure, but finite $\tilde{\alpha}$ limits the retrieval of this structure from the prompt.
In the limit $\tilde{\alpha} \to \infty$, the errors vanish as $\mathcal{E}_{\mathrm{TM}} = \mathcal{E}_{\mathrm{IDG}} = 0$, showing that the model acquires the ability to achieve complete in-distribution generalization once given enough prompt examples, with the transition occurring at $\rho = \kappa$\footnote{The sudden acquisition of this IDG ability may represent a novel type of \textit{emergence}. Unlike standard emergent phenomena in LLMs~\cite{wei2022emergent, raventos2023pretraining}, which are typically driven by increasing dataset size or model scale, this transition is induced by a change in the intrinsic structure of the task distribution itself, as predicted by our theory.}.

\begin{figure}[h]
  \centering
  \includegraphics[width=0.45\textwidth]{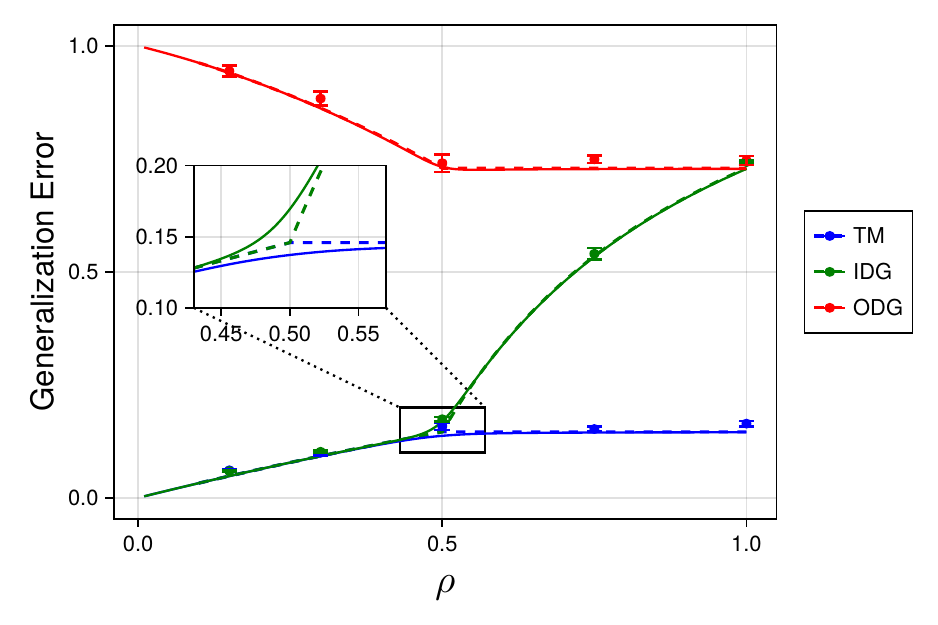}
  \caption{\textbf{Phase transition in the model's capabilities.} 
  Generalization errors ($\mathcal{E}_{\text{TM}}, \mathcal{E}_{\text{IDG}}, \mathcal{E}_{\text{ODG}}$) as a function of the task difficulty $\rho$. 
  The results confirm the predicted phase transition in generalization errors at $\rho = \kappa$.
  The error bars represent the standard errors of the mean over $5$ trials per point. 
  Parameters: $\alpha=200 \text{ (solid lines)}, \alpha\to\infty \text{ (dashed lines)}, \tilde{\alpha}=4.0, \gamma=8.0, \kappa = 0.5, \sigma=0, D = 40$.}\label{fig:plateau}
\end{figure}

\begin{figure*}[ht] 
  \includegraphics[width=1.0\textwidth]{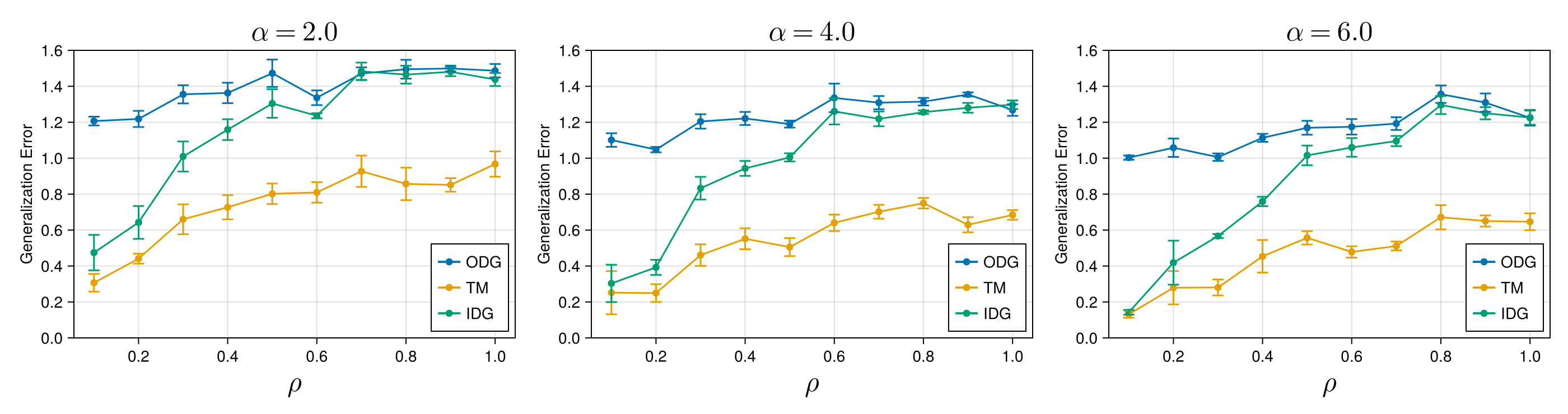}
  \caption{\textbf{Nonlinear softmax attention reproduces the qualitative $\rho$-dependence predicted by the linear theory.} 
  Generalization errors (TM/IDG/ODG) of a one-layer softmax self-attention model as a function of task difficulty $\rho$. 
  The error bars represent the standard errors of the mean over $5$ trials per point. 
  Parameters: $D=20$, $\gamma=3.0$, $\kappa=0.5$, $\sigma=0.01$.}
  \label{fig:experiments}
\end{figure*}

%It is important to note that the IDG error does not exhibit a plateau. Even when the model's intrinsic capabilities are fixed in the $\rho > \kappa$ regime, the evaluation task for IDG becomes intrinsically more difficult as $\rho$ increases, leading to a monotonic increase in error.

\subsection{Implications for Pre-training Strategy}

These theoretically grounded findings reveal that the balance between task difficulty and diversity is a critical lever for tailoring a model's capabilities. 
For instance, developing a specialized model for a specific domain mandates operating in the task-rich regime. This approach, which involves curating a dataset with high diversity of relatively simple tasks, enables the model to efficiently master domain-specific patterns and achieve high TM and IDG performance.

For many practical applications, however, the most effective strategy is to target the ``sweet spot'' near the phase transition ($\rho \approx \kappa$). This regime offers a favorable trade-off while maintaining robust ODG performance, and it avoids the need to prepare unnecessarily difficult training tasks.

Moreover, our findings are consistent with the intuition of curriculum learning~\cite{bengio2009curriculum}, which introduces simple tasks first and gradually shifts to more complex ones. In our framework, the early stage of training corresponds to low task diversity ($\rho > \kappa$), where the model cannot yet learn complex tasks and thus only benefits from simpler ones. As training progresses and task diversity increases, the model enters a regime where it can exploit the richer structure to master more complex and abstract tasks. This explains why a curriculum that progresses from simple to difficult tasks is effective.

\section{Experiments on Non-linear Attention}
\label{sec:experiments}

Our theoretical analysis has so far focused on linear attention, where the predictive behavior and phase transitions can be characterized precisely through macroscopic order parameters. To examine whether the same qualitative trends persist beyond this tractable setting, we also study a one-layer nonlinear model.

Specifically, we replace the linear mechanism in Eq.~\eqref{eq:attention_model} with a one-layer softmax self-attention model, while keeping the pre-training and evaluation protocols identical to those in Sections~\ref{sec:pretraining_phase}, \ref{sec:inference_phase}, and \ref{sec:evaluation_protocol}. 
We train the model on regression tasks with fixed input dimension $D$ and task diversity $\kappa$, and vary the task difficulty $\rho$ to evaluate the resulting generalization performance. The detailed architecture and optimization setup are described in Appendix~\ref{app:experiments}.

Figure~\ref{fig:experiments} shows the TM, IDG, and ODG errors as functions of $\rho$ for the one-layer softmax attention model. 
As $\rho$ decreases, both TM and IDG errors improve systematically, indicating that nonlinear attention can also exploit shared task structures. Moreover, the improvement in TM begins to plateau as $\rho$ approaches $1$; this is consistent with our prediction in Section~\ref{sec:emergence_and_trade_offs_from_varying_task_difficulty}, which suggests that performance is eventually bottlenecked by task diversity ($\kappa$) rather than by the structure of individual tasks. We also note that we do not observe a sharp phase transition, and the ODG error remains higher than the random-guessing baseline in the present setting. These effects are likely to be most visible in the large-$\alpha$ regime, which we do not explore in the present nonlinear experiments.

\section{Conclusion}
In summary, our study sheds light on the mechanisms of ICL in linear transformers. 
We show that ICL functions as a context-dependent noise-reduction mechanism, enabling the model to perform a precise low-rank regression algorithm when the task structure is aligned with the context. 
We further demonstrate that statistical fluctuations in finite training data give rise to an implicit regularization, which stabilizes the learning process. 
Finally, we uncover a sharp phase transition governed by task structure, highlighting a fundamental trade-off between specialization and robustness under distributional shifts. 
These results together provide a coherent framework for understanding how transformers learn to learn.

\section*{Acknowledgements}
KT was supported by JST BOOST NAIS (Grant No.\ JPMJBS2418). 
TT was supported by JSPS KAKENHI (Grant No.\ 23K16960) and JST ACT-X (Grant No.\ JPMJAX24CG). 
TT and YK were supported by JSPS KAKENHI (Grant No.\ 22H05117).

\bibliographystyle{apalike}
\bibliography{ref}

\clearpage
\appendix
\thispagestyle{empty}
\onecolumn
\aistatstitle{Learning Linear Regression with Low-Rank Tasks in-Context:\\ Supplementary Materials}

%\documentclass[twoside]{article}

%\usepackage{aistats2026}
%\usepackage{custom}

%\begin{document}

\section{Summary of Notations}\label{appendix:notations}

In this section, we summarize the notation used in the main text for the reader's convenience.

\begin{longtable}{@{} l p{0.7\textwidth} @{}}
  \caption{Summary of Notation} \label{tab:notation}\\
  \toprule
  \textbf{Symbol} & \textbf{Description} \\
  \midrule
  \endfirsthead
% --- 2ページ目以降のヘッダー ---
\multicolumn{2}{c}%
{{\tablename\ \thetable{} -- continued from previous page}} \\
\toprule
\textbf{Symbol} & \textbf{Description} \\
\midrule
\endhead

% --- 最後のページ以外のフッター ---
\bottomrule
\endfoot

% --- 最後のページのフッター ---
\bottomrule
\endlastfoot

%%%%%%%%%%%%%%%%%%%%%%%%%%%%%%%%%%%%%%%%%%%%%%%%%%%%%%%%%%%%
% Group: Dimensions and Counts
%%%%%%%%%%%%%%%%%%%%%%%%%%%%%%%%%%%%%%%%%%%%%%%%%%%%%%%%%%%%
\multicolumn{2}{@{}l}{\textbf{Dimensions and Counts}} \\
\addlinespace
$D$ & Dimension of the input and task vector space. \\
$r$ & Dimension of the latent space for tasks ($r \le D$). \\
$M$ & Total number of instances in the pre-training set. \\
$M_0$ & Number of base tasks in the initial pool $\mathcal{W}_0$. \\
$L$ & Number of demonstrations (example pairs) in a pre-training context. \\
$\tilde{L}$ & Number of demonstrations in an inference prompt ($\tilde{L} \le L$). \\
\addlinespace

%%%%%%%%%%%%%%%%%%%%%%%%%%%%%%%%%%%%%%%%%%%%%%%%%%%%%%%%%%%%
% Group: Sets, Distributions, and Protocols
%%%%%%%%%%%%%%%%%%%%%%%%%%%%%%%%%%%%%%%%%%%%%%%%%%%%%%%%%%%%
\multicolumn{2}{@{}l}{\textbf{Sets, Distributions, and Protocols}} \\
\addlinespace
$\mathcal{W}_0$ & Base set of $M_0$ task vectors, $\abs{\mathcal{W}_0} = M_0$. \\
$\mathcal{W}$ & Full pre-training set of $M$ tasks, sampled from $\mathcal{W}_0$, $\abs{\mathcal{W}} = M$. \\
$\mathcal{D}_{\mathrm{pretrain}}$ & Generative distribution for pre-training instances. \\
$\mathcal{D}_{\mathrm{test}}$ & Generative distribution for inference instances. \\
$\mathcal{D}_{\mathrm{protocol}}$ & Task distribution for an evaluation protocol (TM, IDG, or ODG). \\
$\mathrm{TM}$ & Task Memorization evaluation protocol. \\
$\mathrm{IDG}$ & In-Distribution Generalization evaluation protocol. \\
$\mathrm{ODG}$ & Out-of-Distribution Generalization evaluation protocol. \\
$\mathcal{E}_{\mathrm{protocol}}$ & Generalization error for an evaluation protocol. \\
\addlinespace

%%%%%%%%%%%%%%%%%%%%%%%%%%%%%%%%%%%%%%%%%%%%%%%%%%%%%%%%%%%%
% Group: Task and Data
%%%%%%%%%%%%%%%%%%%%%%%%%%%%%%%%%%%%%%%%%%%%%%%%%%%%%%%%%%%%
\multicolumn{2}{@{}l}{\textbf{Task and Data}} \\
\addlinespace
$\vb{w}^\mu \in \mathbb{R}^D$ & Ground-truth weight vector defining task $\mu$. \\
$A \in \mathbb{R}^{D \times r}$ & Shared, low-dimensional feature matrix with orthonormal columns. \\
$\vb{v}^\mu \in \mathbb{R}^r$ & Latent feature vector for base task $\mu$. \\
$\vb{x}_l^\mu, \vb{x} \in \mathbb{R}^D$ & Input vectors for demonstrations or queries. \\
$y_l^\mu, y \in \mathbb{R}$ & Output labels for demonstrations or queries. \\
\addlinespace

%%%%%%%%%%%%%%%%%%%%%%%%%%%%%%%%%%%%%%%%%%%%%%%%%%%%%%%%%%%%
% Group: Model Components
%%%%%%%%%%%%%%%%%%%%%%%%%%%%%%%%%%%%%%%%%%%%%%%%%%%%%%%%%%%%
\multicolumn{2}{@{}l}{\textbf{Model Components in the attention model}} \\
\addlinespace
$\Theta = \{V, K, Q \in \mathbb{R}^{D \times D}\}$ & Set of learnable parameters in the linear attention model. \\
$\mathsf{Atten}_{\Theta}$ & The attention function: $\R^{(D+1) \times (L+1)} \to \R^{(D+1) \times (L+1)}$. \\
$\mathsf{Read}$ & The readout function: $\R^{(D+1)\times (L+1)} \to \R$. \\
$C^\mu, P \in \mathbb{R}^{(D+1) \times (L+1)}$ & Feature matrix derived from the context/prompt, the input of the attention model. \\
\addlinespace

%%%%%%%%%%%%%%%%%%%%%%%%%%%%%%%%%%%%%%%%%%%%%%%%%%%%%%%%%%%%
% Group: Model Components
%%%%%%%%%%%%%%%%%%%%%%%%%%%%%%%%%%%%%%%%%%%%%%%%%%%%%%%%%%%%
\multicolumn{2}{@{}l}{\textbf{Model Components in the equivalent linear neural network}} \\
\addlinespace
$\Theta = \{W \in \mathbb{R}^{D \times D}\}$ & Set of learnable parameters in the linear neural network. \\
$H^\mu, H \in \mathbb{R}^{D \times D}$ & Effective feature matrix derived from the context/prompt, the input of the linear neural network. \\
\addlinespace

%%%%%%%%%%%%%%%%%%%%%%%%%%%%%%%%%%%%%%%%%%%%%%%%%%%%%%%%%%%%
% Group: High-Dimensional Analysis Parameters
%%%%%%%%%%%%%%%%%%%%%%%%%%%%%%%%%%%%%%%%%%%%%%%%%%%%%%%%%%%%
\multicolumn{2}{@{}l}{\textbf{High-Dimensional Analysis Parameters}} \\
\addlinespace
$\alpha = L/D$ & Pre-training sample ratio. \\
$\tilde{\alpha} = \tilde{L}/D$ & Inference sample ratio. \\
$\rho = r/D$ & Task subspace dimensionality. \\
$\kappa = M_0 / D$ & Task diversity. \\
$\gamma = M / (D M_0)$ & Training data density. \\

\end{longtable}

\section{Justification of Equivalence of the Linear Attention Model to a Linear Neural Network}
\label{app:eq-lenear-nn}

In this section, we demonstrate that the single-layer linear attention model employed in our study can be mapped to an equivalent single-layer linear neural network. 
This equivalence provides a tractable framework for our theoretical analysis.

First, from the definition of the attention function $\mathsf{Atten}_{\Theta}$, we have:
\begin{align}
  \mathsf{Atten}_{\Theta} (C) &= C + \frac{1}{L} VC (KC)^\top (QC) \\
  &= C + \frac{1}{L} VC C^\top (K^\top Q) C 
\end{align}
The dependency on the key and query matrices, $K$ and $Q$, occurs only through the product $K^\top Q$. We can therefore simplify the parameterization by defining a single matrix $R = K^\top Q \in \mathbb{R}^{(D+1) \times (D+1)}$.

To analyze the interaction between features and labels, we partition the parameter matrices $V$ and $R$ into blocks corresponding to the feature dimensions ($D$) and the label dimension (1):
\begin{align}
  R = \pmqty{R_{11} & \vb{r}_{12} \\ \vb{r}_{21}^\top & r_{22}}, \quad V = \pmqty{V_{11} & \vb{v}_{12} \\ \vb{v}_{21}^\top & v_{22}},
\end{align}
where $R_{11}, V_{11}\in \mathbb{R}^{D\times D}$, $\vb{r}_{12}, \vb{v}_{12}, \vb{r}_{21}, \vb{v}_{21}\in \mathbb{R}^{D}$ and $r_{22}, v_{22}\in \mathbb{R}$.

In this notation, a direct calculation of the model's prediction, which is the scalar value at the query position $\{ \dots \}_{D+1, L+1}$, yields:
\begin{align}
  \mathsf{Read}\ab(\mathsf{Atten}_{\Theta} (C^\mu) )&= \ab\{C^\mu + \frac{1}{L} VC^\mu{C^\mu}^\top R C^\mu\}_{D+1, L+1} \\
  &= \frac{1}{L} \ab(\vb{v}_{21}^\top X^\mu{X^\mu}^\top R_{11} + v_{22}{\vb{y}^\mu}^\top {X^\mu}^\top R_{11} + \vb{v}_{21}^\top {X^\mu}^\top\vb{y}^\mu \vb{r}_{21} + \vb{r}_{21}{\vb{y}^\mu}^\top\vb{y}^\mu ) \vb{x}_{L+1}^\mu,
\end{align}
where $X^\mu = \pmqty{\vb{x}_1^\mu & \vb{x}_2^\mu & \cdots & \vb{x}_L^\mu} \in \mathbb{R}^{D \times L}$ and $\vb{y}^\mu = \pmqty{y_1^\mu & y_2^\mu & \cdots & y_{L}^\mu & 0}^\top \in \mathbb{R}^{L+1}$.

To make the model analytically tractable, we follow recent works~\cite{Wu2023-ul, Frei2024-fe, zhang2025training} and introduce a key simplification by setting the off-diagonal block matrices to zero: $\vb{v}_{21} = \mathbf{0}$ and $\vb{r}_{21} = \mathbf{0}$. 
This simplification is well-founded, as it has been shown that for networks initialized with these parameters at zero, they remain zero throughout training under gradient flow dynamics~\cite{zhang2025training}. 
While this reduces complexity, the resulting architecture is still sufficiently rich to exhibit nontrivial in-context learning phenomena.

Under this assumption, the prediction formula simplifies dramatically. By defining an effective weight matrix $W = \ab(v_{22} R_{11})^\top \in \mathbb{R}^{D \times D}$, the model's output becomes:
\begin{equation}
  \mathsf{Read}\ab(\mathsf{Atten}_{\Theta} (C^\mu) ) = \frac{1}{L} v_{22}{\vb{y}^\mu}^\top {X^\mu}^\top R_{11} \vb{x}_{L+1}^\mu = \tr \ab(W {H^\mu}^\top )
\end{equation}
This final expression demonstrates that the attention model's output is equivalent to that of a single-layer linear neural network. 
This equivalent network maps an effective feature matrix $H^\mu$ to a scalar prediction via the trace operator, 
corresponding to a mapping from $\mathbb{R}^{D \times D} \to \mathbb{R}$.

\section{Justification of Equivalence of the Linear Attention Model to a Many-Teacher-Student Model}
\label{app:eq-ts}

To enable a tractable high-dimensional analysis using the replica method, we map our ICL model to an equivalent teacher-student framework with multiple teachers. This reformulation is not only a crucial technical step, but it also provides a clear interpretation of the pre-training phase. Furthermore, we believe this analytical strategy can serve as a valuable starting point for studying a broader class of similar ICL models. The core of this equivalence is the following statement:

\begin{statement}\label{statement:multi_teacher_student_framework}
The macroscopic statistical quantities of the learning process based on Eq.~\eqref{eq:equivalent_NN} remain unchanged if the true training labels $y^\mu_{L+1}$ are replaced by the outputs of a task-specific teacher model. The teacher for task $\mu$ is defined by the weight matrix $W_{\mathsf{teacher}}^\mu = \vb{w}^\mu (\vb{w}^\mu)^\top / D$, leading to the substitution:
\begin{equation}
y^\mu_{L+1} \to \tr \ab(W_{\mathsf{teacher}}^\mu H^\mu) + \epsilon^\mu .
\end{equation}
\end{statement}
For the justification of Statement~\ref{statement:multi_teacher_student_framework} and a discussion about its relevance to the traditional teacher-student model, see the following subsections.

This equivalence is based on an analogous to the Gaussian Equivalence Theorem~\cite{mei2022generalization}, as the first two moments of the true labels and the teacher's outputs are identical, while higher-order correlations become negligible in the high-dimensional limit.

The student model in our framework is not merely memorizing the solutions to specific tasks. Because it learns from an ensemble of $M$ different teachers, one for each training instance, it is forced to distill the underlying, universal problem-solving algorithm common to all of them.

\subsection{Justification of Equivalence of the Linear Attention Model to a Many-Teacher-Student Model}
In this subsection, we will justify the equivalence to a many-teacher-student model, by checking the consistency of the first and second-order statistics.
Specifically, we will show the following lemma:

\begin{lemma}\label{lemma:ts}
  Let $W_{\mathsf{teacher}}^\mu = \vb{w}^\mu (\vb{w}^\mu)^\top / D$ be the weight matrix of a task-specific teacher model, 
  and $\bar{y}^\mu_{L+1} = \tr \ab(W_{\mathsf{teacher}}^\mu H^\mu) + \epsilon^\mu$ be the output of a task-specific teacher model. 
  Then, for the fixed $\mathcal{W}$, the first and second-order statistics of the teacher model are consistent with the ICL model, i.e.,
  \begin{align}
    \E \ab[\bar{y}^\mu_{L+1}] &= \E \ab[{y}^\mu_{L+1}] = 0\\
    \E \ab[\bar{y}^\mu_{L+1}\bar{y}^\nu_{L+1}] &= \E \ab[{y}^\mu_{L+1}{y}^\nu_{L+1}]  = \delta_{\mu \nu} \ab(1+\sigma^2) + \mathcal{O}\ab(\frac{1}{D}) \\
    \E \ab[\bar{y}^\mu_{L+1} H^\nu_{ij}] &= \E \ab[{y}^\mu_{L+1} H^\nu_{ij}] = \frac{1}{D} \delta_{\mu \nu} w_i^\mu w_j^\mu + \mathcal{O}\ab(\frac{1}{D^2}) .
  \end{align}
\end{lemma}

\begin{proof}
  Due to $H^\mu_{ij}=0$, we have $\E \ab[\bar{y}^\mu_{L+1}] = \E \ab[{y}^\mu_{L+1}] = 0$.
  Also, straightforwardly we have $\E\ab[y^\mu_{L+1} y^\nu_{L+1}] = \delta_{\mu \nu} \ab(1+\sigma^2)$ and 
  \begin{align}
    \E \ab[H_{ij}^\mu y_{L+1}^\nu] &= \frac{D}{L}\delta_{\mu \nu} x_{L+1, i}^\mu \ab(\sum_{l=1}^L y_l^\mu x_{l, j}^\mu) \ab(\sum_{k=1}^D w_k^\mu x_{l, k}^\mu + \epsilon_{l}^\mu) \\
    &= \frac{D}{L} \delta_{\mu \nu} \sum_{k, k'} w^\mu_k w^\mu_{k'} \E \ab[x_{L+1, i}^\mu x_{L+1, k'} \sum_{l=1}^L x_{l, j}^\mu x_{l, k}^\mu] \\
    &= \frac{1}{D} \delta_{\mu \nu} w_i^\mu w_j^\mu .
  \end{align}
  For the second-order statistics of the teacher model, we have:
  \begin{align}
    \E \ab[H^\mu_{ij} H^\nu_{i' j'}] &=  \E\ab[ \frac{D^2}{L^2} \ab[x^\mu_{L+1, i} \sum_{l=1}^L \ab(\sum_{k=1}^D w_k^\mu x_{l, k}^\mu + \epsilon_{l,k}^\mu)x_{l, j}^\mu]\ab[x^\nu_{L+1, i'} \sum_{l'=1}^L \ab(\sum_{k'=1}^D w_{k'}^\nu x_{l', k'}^\nu + \epsilon_{l'}^\nu)x_{l', j'}^\nu]] \\
    &= \delta_{\mu \nu}  \E \left[\frac{D^2}{L^2} \sum_{k, k'} w_k^\mu w_{k'}^\nu \ab[x^\mu_{L+1, i} x^\nu_{L+1, i'} \ab(\sum_{l=1}^L \sum_{l'=1}^L x_{l, k}^l x_{l, j}^l x_{l', k'}^l x_{l', j'}^l)] \right. \notag \\
    &\hspace{20em}  \left. + \frac{D^2}{L^2} x^\mu_{L+1, i} x^\mu_{L+1, i'} \sum_{l=1}^L \sum_{l'=1}^L x_{l, j}^l x_{l', j'}^l \epsilon_{l}^l \epsilon_{l'}^l \right] \\
    &= \delta_{\mu \nu} \frac{\delta_{ii'}}{D} \ab(w_jw_{j'} + \frac{D}{L}\ab(1+\sigma^2)\delta_{jj'}) + \mathcal{O}\ab(\frac{1}{D^2}) \\
    \E\ab[\bar{y}^\mu_{L+1} H^\nu_{ij}] &= \frac{1}{D} \sum_{i', j'} w_{i'}^\mu w_{j'}^\nu \E \ab[H_{ij}^\mu H_{i'j'}^\nu] \\
    &= \frac{1}{D} \sum_{i', j'} w_{i'}^\mu w_{j'}^\nu \delta_{\mu \nu} \frac{\delta_{ii'}}{D} \ab(w_jw_{j'} + \frac{D}{L}\ab(1+\sigma^2)\delta_{jj'}) + \mathcal{O}\ab(\frac{1}{D^2}) \\
    &= \frac{1}{D} \delta_{\mu \nu} w_i^\mu w_j^\mu + \mathcal{O}\ab(\frac{1}{D^2}) \\
    \E \ab[\bar{y}^\mu_{L+1} \bar{y}^\nu_{L+1}] &= \frac{1}{D^2} \sum_{i,j}\sum_{i',j'} w_i^\mu w_j^\nu w_{i'}^\mu w_{j'}^\nu \E \ab[H_{ij}^\mu H_{i'j'}^\nu] \\
    &= \frac{1}{D^2} \sum_{i,j}\sum_{i',j'} w_i^\mu w_j^\nu w_{i'}^\mu w_{j'}^\nu \delta_{\mu \nu} \frac{\delta_{ii'}}{D} \ab(w_jw_{j'} + \frac{D}{L}\ab(1+\sigma^2)\delta_{jj'}) + \mathcal{O}\ab(\frac{1}{D^2}) \\
    &= \delta_{\mu \nu} \ab(1+\sigma^2) + \mathcal{O}\ab(\frac{1}{D}),
  \end{align}
  which completes the proof.
\end{proof}

\subsection{Relation to Committee Machine}
This process reveals a crucial distinction from conventional multi-teacher models like committee machines, which learn the average behavior of a teacher ensemble to capture their consensus. 
Our student, in contrast, is not trained to find this middle ground but must learn to imitate every individual teacher. 
Statement~\ref{lemma:ts} reveals that the essence of ICL is not knowledge aggregation, but the acquisition of a universal meta-algorithm capable of acting as any specialized teacher based on the context.

\section{Preliminary Lemmas in Random Matrix Theory}\label{appendix:preliminary_lemmas}
In this section, we provide the proofs for several technical lemmas required in the calculation of the replica method in the main text. These lemmas establish fundamental properties of random matrices that arise in our analysis, particularly concerning resolvent functions and their derivatives.

\subsection{Resolvent Functions for Wishart-Type Matrices}

We begin by establishing the resolvent functions for a class of random matrices that play a central role in our replica calculation.

\begin{lemma}\label{lemma:resolvent_S}
  Consider the random matrix
  \begin{align}
    \tilde{S} &= \frac{D}{r} \frac{1}{M_0}  \sum_{\mu=1}^{M_0} \vb{v}^\mu {\vb{v}^\mu}^\top \in \R^{r \times r} \\
    S &= \frac{D}{r} \frac{1}{M_0} \sum_{\mu=1}^{M_0} A\vb{v}^\mu {\vb{v}^\mu}^\top A^\top \in \R^{D \times D},
  \end{align}
  where $0<r<D$, $\vb{v}^\mu \in \R^r$ are random standard normal vectors and $A \in \R^{D \times r}$ is a random orthogonal basis.
  Then, the resolvent of $S$ and $\tilde{S}$ are given by
  \begin{align}
    g_{\tilde{S}}(z) &=\frac{1}{D} \tr \ab(\ab(\tilde{S}-zI_D)^{-1}) = \frac{-(\kappa \rho z + \rho - \kappa) - \sqrt{\left(\kappa \rho z + \rho - \kappa \right)^2 - 4\rho^2 \kappa z}}{2\rho z} \label{eq:g_tilde_S} \\
    g_S(z) &=\frac{1}{D} \tr \ab(\ab({S}-zI_D)^{-1}) =  \frac{\kappa + \rho - \kappa \rho z - 2 - \sqrt{\left(\rho\kappa z + \rho - \kappa \right)^2 - 4\rho^2 \kappa z}}{2z} \label{eq:g_S}
  \end{align}
  for $z < 0$ respectively, where $\rho = r/D$, $\kappa = M_0/D$ and $D, r, M_0 \to \infty$.
\end{lemma}

\begin{proof}
  The random matrix $\tilde{S}$ is classified as a Wishart matrix (or a scaled sample covariance matrix). 
  The empirical spectral distribution of such a matrix is known to converge to the Marchenko-Pastur distribution in the given asymptotic limit. 
  Therefore, Eq.~\eqref{eq:g_tilde_S}, which represents its Stieltjes transform, follows directly from this standard result in random matrix theory (e.g., see \cite{potters2020first}).

  For the random matrix $S$, let $P_S(\lambda)$ and $P_{\tilde{S}}(\lambda)$ denote the characteristic polynomials of $S$ and $\tilde{S}$, respectively.
  Then we have
  \begin{align}
    P_{S}(\lambda) &= \det\ab(\lambda I_D - S) = \lambda^{D-r} \det\ab(\lambda I_r - \tilde{S}) = \lambda^{D-r} P_{\tilde{S}}(\lambda) 
  \end{align}
  which implies that the non-zero eigen values of $S$ and $\tilde{S}$ coincide, including multiplicities.
  Therefore,
  \begin{align}
    g_{S}(z) &= - \frac{D-r}{D} \frac{1}{z} + \frac{r}{D} g_{\tilde{S}}(z) = - \ab(1-\rho)\frac{1}{z} + \rho g_{\tilde{S}}(z) \\
    &=\frac{1}{D} \tr \ab(\ab({S}-zI_D)^{-1}) =  \frac{\kappa + \rho - \kappa \rho z - 2 - \sqrt{\left(\rho\kappa z + \rho - \kappa \right)^2 - 4\rho^2 \kappa z}}{2z} 
  \end{align}
  as desired in Eq.~\eqref{eq:g_S}.
\end{proof}

\subsection{Recurrence Relations for Matrix Moments}

Having established the resolvent functions, we now derive recurrence relations that allow us to compute higher-order moments involving powers of the random matrices and their resolvents. 
These relations are essential for the subsequent calculations in the saddle-point equations.

\begin{lemma}\label{lemma:derivation_of_E}
  Let $\mathcal{M} = a S + b I_D \in \R^{D \times D}$, where $a, b>0$ and $S$ is defined in Lemma~\ref{lemma:resolvent_S}. Define $E_{n, m}$ by
  \begin{align}
    E_{n, m} = \frac{1}{D} \E \ab[\tr S^n \mathcal{M}^{-m} ] = \frac{1}{D} \E \ab[\tr S^n (a S + b I_D)^{-m} ]
  \end{align}
  for integer $n, m \geq 0$. Then, the following recurrence relation holds:
  \begin{align}
    E_{0, m} &= \frac{1}{b^m} \frac{1}{(m-1)!} g_S^{(m-1)}(z) \quad (\text{for } m \ge 1) \\
    E_{n, 0} &= 1 \quad (n \geq 1) \\
    E_{n, m} &= \frac{1}{b} E_{n-1, m-1} + z E_{n-1, m} \quad (n \geq 1, m \geq 1) ,
  \end{align}
  where $z = -b/a$ and $g^{(m-1)}_S(z)$ denotes the $(m-1)$-th derivative of $g_S(z)$.
\end{lemma}

\begin{proof}
  The inverse powers of $\mathcal{M}$ are given by:
  \begin{equation}
    \mathcal{M}^{-m} = \frac{1}{b^m}(S - zI_D)^{-m}
  \end{equation}
  We use the standard definition of the Stieltjes transform, $g_S(z) = \frac{1}{D} \mathbb{E}[\text{tr}((S-zI_D)^{-1})]$.
  The $k$-th derivative of the resolvent $(S-zI_D)^{-1}$ with respect to $z$ is:
  \begin{equation}
  \frac{d^k}{dz^k} (S - zI_D)^{-1} = k! (S - zI_D)^{-(k+1)}
  \end{equation}
  From this, we can express the $m$-th power of the resolvent as:
  \begin{equation}
  (S - zI_D)^{-m} = \frac{1}{(m-1)!} \frac{d^{m-1}}{dz^{m-1}} (S - zI_D)^{-1}
  \end{equation}
  The $k$-th derivative of the Stieltjes transform is therefore:
  \begin{equation}
  g_S^{(k)}(z) = \frac{d^k g_S(z)}{dz^k} = \frac{1}{D} \mathbb{E}\left[\tr \left(\frac{d^k}{dz^k}(S - zI_D)^{-1}\right)\right] = \frac{k!}{D} \mathbb{E}\left[\text{tr}\left((S - zI_D)^{-(k+1)}\right)\right]
  \end{equation}
  
  First, we prove the expression for $E_{0, m}$. By definition, for $m \ge 1$:
  \begin{align}
  E_{0, m} &= \frac{1}{D} \mathbb{E}\left[\text{tr}(\mathcal{M}^{-m})\right] \\
  &= \frac{1}{D} \mathbb{E}\left[\text{tr}\left(\frac{1}{b^m}(S - zI_D)^{-m}\right)\right] \\
  &= \frac{1}{b^m} \frac{1}{D} \mathbb{E}\left[\tr \left( \frac{1}{(m-1)!} \frac{d^{m-1}}{dz^{m-1}} (S - zI_D)^{-1} \right)\right] \\
  &= \frac{1}{b^m} \frac{g_S^{(m-1)}(z)}{(m-1)!}
  \end{align}
  
  Second, the expression for $E_{n, 0}$ for $n \ge 1$ follows directly from the definition:
  $$
  E_{n, 0} = \frac{1}{D} \mathbb{E}\left[\tr (S^n \mathcal{M}^0)\right] = \frac{1}{D} \mathbb{E}\left[\text{tr}(S^n)\right] = 1
  $$

  Third, we prove the recurrence relation for $n \ge 1$ and $m \ge 1$. We start from the definition of $E_{n, m}$ and use the identity $S^n = S^{n-1}((S - zI_D) + zI_D)$.
  \begin{align}
  E_{n, m} &= \frac{1}{D} \mathbb{E}\left[\tr (S^n \mathcal{M}^{-m})\right] \\
  &= \frac{1}{D} \mathbb{E}\left[\tr \left( \left( S^{n-1}(S-zI_D) + zS^{n-1} \right) \mathcal{M}^{-m} \right)\right] \\
  &= \frac{1}{D} \mathbb{E}\left[\tr \left( S^{n-1}(S-zI_D) \mathcal{M}^{-m} \right)\right] + \frac{1}{D} \mathbb{E}\left[\text{tr}\left( zS^{n-1} \mathcal{M}^{-m} \right)\right] \label{eq:E_n_m} .
  \end{align}
  For the first term of Eq.~\eqref{eq:E_n_m}, we substitute $(S - zI_D) = \frac{1}{b} \mathcal{M}$:
  \begin{equation}
  \frac{1}{D} \mathbb{E}\left[\text{tr}\left( S^{n-1}(S-zI_D) \mathcal{M}^{-m} \right)\right] = \frac{1}{b} \frac{1}{D} \mathbb{E}\left[\text{tr}\left( S^{n-1} \mathcal{M}^{-(m-1)} \right)\right] = \frac{1}{b} E_{n-1, m-1} .
  \end{equation}
  The second term is:
  \begin{equation}
  \frac{1}{D} \mathbb{E}\left[\text{tr}\left( zS^{n-1} \mathcal{M}^{-m} \right)\right] = z \left( \frac{1}{D} \mathbb{E}\left[\text{tr}\left( S^{n-1} \mathcal{M}^{-m} \right)\right] \right) = z E_{n-1, m} .
  \end{equation}
  Combining the two terms gives the recurrence relation:
  \begin{equation}
  E_{n, m} = \frac{1}{b} E_{n-1, m-1} + z E_{n-1, m},
  \end{equation}
  which completes the proof.
\end{proof}

\subsection{Explicit Expressions for Matrix Moments}

The recurrence relations established in the previous lemma can be used to derive explicit expressions for specific combinations of matrix powers and resolvents that frequently appear in our calculations.

%By applying the recurrence relation in this Lemma~\ref{lemma:derivation_of_E}, we immediately obtain the following identities.

\begin{proposition}\label{lemma:explicit_expressions_for_matrix_moments}
  Let $S$, $\mathcal{M}$ and $z$ be as defined in Lemma~\ref{lemma:resolvent_S} and \ref{lemma:derivation_of_E}. Then, the following relations hold:
  \begin{align}
   \frac{1}{D} \E \ab[\tr \mathcal{M}^{-1} ] &= \frac{g_S(z)}{b} \label{eq:E_M_1}\\
   \frac{1}{D} \E \ab[\tr \mathcal{M}^{-2} ] &= \frac{g_S'(z)}{b^2} \label{eq:E_M_2}\\
   \frac{1}{D} \E \ab[\tr S\mathcal{M}^{-1} ] &= \frac{1+zg_S(z)}{b} \label{eq:E_M_3}\\
   \frac{1}{D} \E \ab[\tr S\mathcal{M}^{-2} ]  &= \frac{g_S(z) + zg_S'(z)}{b^2} \label{eq:E_M_4}\\
   \frac{1}{D} \E \ab[\tr S^2 \mathcal{M}^{-1} ] &= \frac{1+z+z^2g_S(z)}{b} \label{eq:E_M_5}\\
   \frac{1}{D} \E \ab[\tr S^2 \mathcal{M}^{-2} ] &= \frac{1+2zg_S(z) + z^2g_S'(z)}{b^2} \label{eq:E_M_6}\\
   \frac{1}{D} \E \ab[\tr S^3 \mathcal{M}^{-2} ] &= \frac{1+2z+3z^2g_S(z) + z^3g_S'(z)}{b^2} \label{eq:E_M_7}  .
  \end{align}
\end{proposition}

\begin{proof}
  The proof follows directly from Lemma~\ref{lemma:derivation_of_E}.
\end{proof}

\subsection{Trace Identity for Projected Matrices}

Finally, we establish a trace identity that relates the trace of a projected matrix to the trace of its lower-dimensional counterpart. 

\begin{lemma}\label{lemma:trace_identity}
  Let $A$, $S$ and $\tilde{S}$ be as defined in Lemma~\ref{lemma:resolvent_S}. 
  For scalars $a, b, l > 0$, the following identity holds:
  \begin{equation}
  \tr\left(AA^\top (aS+bI_D)^{-l}\right) = \tr\left((a\tilde{S} + bI_r)^{-l}\right)
  \end{equation}
\end{lemma}
  
\begin{proof}
  We prove this identity by performing a change of basis. 
  Since the $r$ columns of $A$ are orthonormal, we can extend this set to form a complete orthonormal basis for $\mathbb{R}^D$. 
  Let $B \in \mathbb{R}^{D \times (D-r)}$ be a matrix whose columns form an orthonormal basis for the orthogonal complement of the column space of $A$. 
  The resulting matrix $U = [A \, B] \in \R^{D \times D}$ is an orthogonal matrix, satisfying $U^\top U = UU^\top = I_D$.
  
  The condition $AA^\top S = S$ implies that the column space of $S$ is contained within that of $A$. 
  Consequently, $S$ must annihilate any vector orthogonal to the column space of $A$. 
  This means $B^\top S = 0$, which can be verified as follows:
  \begin{equation}
      B^\top S = B^\top (AA^\top S) = (B^\top A)A^\top S = 0.
  \end{equation}
  In the new basis defined by $U$, the matrix $S$ is represented by $U^\top S U$. This similarity transformation reveals a block upper-triangular structure:
  \begin{equation}
      U^\top S U = \begin{pmatrix} A^\top  B^\top \end{pmatrix} S \begin{pmatrix} A & B \end{pmatrix} = \begin{pmatrix} A^\top S A & A^\top S B \\ B^\top S A & B^\top S B \end{pmatrix} = \begin{pmatrix} A^\top SA & A^\top SB \\ 0 & 0 \end{pmatrix}. \label{eq:USU}
  \end{equation}
  Now, we express the left-hand side (LHS) of Eq.~\eqref{eq:USU} in this basis. Using the cyclic property of the trace, $\tr(X) = \tr(U^\top X U)$, we have:
  \begin{equation}
      \mathrm{LHS} = \tr\left(U^\top (AA^\top) U  U^\top (aS+bI_D)^{-l} U\right).
  \end{equation}
  The projection matrix $AA^\top$ and the term $(aS+bI_D)$ transform as:
  \begin{align}
      U^\top AA^\top U &= \begin{pmatrix} I_r & 0 \\ 0 & 0 \end{pmatrix} \\
      U^\top(aS+bI_D)U &= a(U^\top S U) + b I_D = \begin{pmatrix} aA^\top SA+bI_r & aA^\top SB \\ 0 & bI_{D-r} \end{pmatrix}.
  \end{align}
  The inverse of a block upper-triangular matrix is also block upper-triangular, and squaring it preserves this structure. We only need the diagonal blocks for the trace calculation:
  \begin{equation}
      \left(U^\top(aS+bI_D)U\right)^{-l} = \begin{pmatrix} (aA^\top SA+bI_r)^{-l} & * \\ 0 & (bI_{D-r})^{-l} \end{pmatrix},
  \end{equation}
  where $*$ denotes the off-diagonal block, which is irrelevant for the trace. Substituting these into the expression for the LHS:
  \begin{align}
      \mathrm{LHS} &= \tr \left( \begin{pmatrix} I_r & 0 \\ 0 & 0 \end{pmatrix} \begin{pmatrix} (aA^\top SA+bI_r)^{-l} & * \\ 0 & b^{-l}I_{D-r} \end{pmatrix} \right)  \\
      &= \tr  \begin{pmatrix} (aA^\top SA+bI_r)^{-l} & * \\ 0 & 0 \end{pmatrix} \\
      &= \tr \ab((aA^\top SA+bI_r)^{-l})
  \end{align}
  This is identical to the right-hand side of the proposition, which completes the proof.
\end{proof}

\vfill
\section{Replica Calculation}\label{appendix:replica}

\subsection{Results with Complete Statements}

In this subsection, we present the main results (Result~ \ref{result:decomposition}, \ref{result:asump_decomp}, \ref{result:implicit_regularization}, \ref{result:asymptotic_coefficients}, \ref{result:gen_error_simple}) with complete statements, including the expression for the generalization error, which are derived using the replica method.
We analyze a more general model that incorporates a regularization term into the loss function. The results for the unregularized model can be recovered by simply setting the regularization strength $\lambda = 0$.

\begin{result}(Complete Statement of Optimal Parameter Matrix)\label{result:complete_statement_W}
  Let the optimal parameter matrix $W^* \in \mathbb{R}^{D \times D}$ be the solution that minimizes the cost function $\mathcal{L}(W)$, defined as:
  \begin{equation}
      \mathcal{L}(W) = \frac{1}{2} \sum_{\mu=1}^M \ab( \sum_{i=1}^D\sum_{j=1}^D \ab( \frac{w_i^\mu w_j^\mu}{D} - W_{ij} ) H_{ij}^\mu)^2 + \frac{M_0 \lambda}{2} \sum_{i=1}^D\sum_{j=1}^D W_{ij}^2
  \end{equation}
  For a given fixed set of tasks $\mathcal{W}_0 = \{\vb{w}^1, \dots, \vb{w}^{M_0}\}$, there exist non-negative scalar constants $\hat{q}, \hat{m}, \hat{\chi}, \hat{\bar{q}}, \hat{\bar{\chi}}$ such that the optimal matrix $W^*$ is given by the expression:
  \begin{align}
      W^* 
      &\overset{\mathrm{d}}{=} \argmin_{W}\ab[ \frac{\lambda + \hat{\bar{q}}}{2}\tr\ab(WW^\top) +\frac{\hat{q}}{2} \tr \ab(WSW^\top) - \tr\ab(\ab(\sqrt{\hat{\chi}}T + \sqrt{\hat{\bar{\chi}}}R + \hat{m}S)W^\top) ] \\
      &\overset{\mathrm{d}}{=} \ab(\sqrt{\hat{\chi}}T + \sqrt{\hat{\bar{\chi}}}R + \hat{m}S) \ab(\hat{q}S + \ab(\lambda + \hat{\bar{q}})I_D)^{-1},
  \end{align}
  where the matrices $T, R, S \in \R^{D \times D}$ are defined as follows:
  \begin{equation}
      T = \frac{1}{M_0}\sum_{\mu=1}^{M_0} \vb{\xi}{\vb{w}^\mu}^\top, \quad R = \ab(R_{ij})_{i,j=1}^D \sim \mathcal{N}\ab(0, \frac{1}{M_0}), \quad S = \frac{1}{M_0} \sum_{\mu=1}^{M_0} \vb{w}^\mu{\vb{w}^\mu}^\top,
  \end{equation}
  with $\vb{\xi} \in \R^{D} \sim \mathcal{N}\ab(0, I_D)$.
\end{result}

\begin{result}(Parameter Derivation)
  \label{result:parameter_derivation}
  The order parameters $q, m, q_0, m_0, \bar{q}, \bar{m}$ (with auxiliary parameters $\chi$) in Result.~\ref{result:complete_statement_W} 
  and the conjugate parameters $\hat{q}, \hat{m}, \hat{\chi}, \hat{\bar{q}}, \hat{\bar{\chi}}$ in Result.~\ref{result:complete_statement_E} are given by the solution of the following system of equations:
  \begin{align}
    q &= \frac{1}{\kappa\hat{q}^2}\ab[ \kappa\hat{m}^2 \ab(1+2z+3z^2g_S(z) + z^3g_S'(z)) + \hat{\bar{\chi}}\ab(g_S(z) + zg_S'(z)) + \hat{\chi}\ab(1+2zg_S(z)+z^2g_S'(z))] \label{eq:q_eq} \\ 
    \bar{q} &= \frac{1}{\kappa\hat{q}^2} \ab[ \kappa\hat{m}^2 \ab(1+2zg_S(z) + z^2g'_S(z)) + \hat{\bar{\chi}} g'_S(z) + \hat{\chi} \ab(g_S(z) + z g'_S(z))] \label{eq:bar_q_eq} \\ 
    q_0 &= \frac{1}{\rho\kappa\hat{q}^2} \ab[ \kappa\hat{m}^2\ab(1+2zg_S(z) + z^2g'_S(z)) + \rho \hat{\bar{\chi}} g'_{\tilde{S}}(z) + \hat{\chi} \ab(g_S(z) + z g'_S(z))] \label{eq:q_0_eq} \\ 
    m &= \frac{\hat{m}}{\hat{q}} \ab(1+z + z^2g_S(z)) \label{eq:m_eq} \\ 
    \bar{m}& = \frac{\hat{m}}{\hat{q}} \ab(1+zg_S(z)) \label{eq:bar_m_eq} \\ 
    m_0 &= \frac{1}{\rho} \bar{m} \label{eq:m_0_eq} \\ 
    \chi &= \frac{1}{\kappa\hat{q}} \ab(1+zg_S(z)) \label{eq:chi_eq} \\ 
    \bar{\chi}& = \frac{1}{\kappa \hat{q}} g_S(z) \label{eq:bar_chi_eq} , 
   \end{align}
   where $z = - \frac{\hat{\bar{q}} + \lambda}{\hat{q}}$ and 
   \begin{align}
    \hat{q} &= \gamma \frac{1}{1 +\chi +\frac{1}{\alpha}\ab(1+\sigma^2)\bar{\chi}} \\
    \hat{\bar{q}} &= \gamma \frac{\frac{1}{\alpha}\ab(1+\sigma^2)}{1 +\chi +\frac{1}{\alpha}\ab(1+\sigma^2)\bar{\chi}} \\
    \hat{m} &= \gamma  \frac{1}{1 +\chi +\frac{1}{\alpha} \ab(1+\sigma^2) \bar{\chi}} \\ 
    \hat{\chi} &= \gamma  \frac{\frac{1}{\alpha}\ab(1+\sigma^2)\bar{q} + q - 2m + 1 + \sigma^2}{\ab(1 + \chi + \frac{1}{\alpha} \ab(1+\sigma^2)\bar{\chi})^2} \\
    \hat{\bar{\chi}} &= \gamma \frac{1}{\alpha}\ab(1+\sigma^2)\frac{\frac{1}{\alpha}\ab(1+\sigma^2)\bar{q} + q - 2m + 1 + \sigma^2}{\ab(1 + \chi + \frac{1}{\alpha} \ab(1+\sigma^2)\bar{\chi})^2} .
  \end{align}
  Here, we define $z = - \frac{\hat{\bar{q}} + \lambda}{\hat{q}}$ and the resolvents $g_S(z)$ and $g_{\tilde{S}}(z)$ are defined in Lemma~\ref{lemma:resolvent_S}.
\end{result}

\begin{result}(Complete Statement of Generalization Error)\label{result:complete_statement_E}
  There exist non-negative scalar constants $q, m, q_0, m_0, \bar{q}, \bar{m}$ such that the generalization error of the optimal parameter matrix $W^*$ is given by:
  \begin{align}
    \mathcal{E}_{\mathrm{TM}}&= 1 - 2m + q + \frac{\bar{q}}{\tilde{\alpha}} \\
    \mathcal{E}_{\mathrm{IDG}}&= 1 - 2{m}_0 + {q}_0 + \frac{\bar{q}}{\tilde{\alpha}} \\
    \mathcal{E}_{\mathrm{ODG}}&= 1 - 2{\bar{m}} + {\bar{q}} + \frac{\bar{q}}{\tilde{\alpha}} .
  \end{align}
\end{result}

\begin{proposition}
\label{prop:asympt}
  Under the condition that $\kappa \gamma > 1$, the order parameters in $\alpha \gg 1$ are given by:
  \begin{align}
    m &= 1 + \mathcal{O}\ab(\frac{1}{\alpha}), & \hat{m} &=  \gamma - \frac{1}{\kappa} + \mathcal{O}\ab(\frac{1}{\alpha})\\
    \bar{m} &= \min\ab(\kappa, \rho) + \mathcal{O}\ab(\frac{1}{\alpha})& \\
    m_0 &= \frac{1}{\rho} \min \ab(\kappa, \rho)+ \mathcal{O}\ab(\frac{1}{\alpha}) \\
    q &= 1 + \frac{\sigma^2}{\kappa\gamma-1} \min\ab(\kappa, \rho)+ \mathcal{O}\ab(\frac{1}{\alpha}) & \hat{q} &=  \gamma -\frac{1}{\kappa} + \mathcal{O}\ab(\frac{1}{\alpha}) \\
    \bar{q} &= \frac{\sigma^2}{1+\sigma^2} \frac{1-\min\ab(\kappa, \rho)}{\kappa\gamma-1} \alpha  + \min(\kappa, \rho)\ab[1+\frac{1-\min\ab(\kappa, \rho)}{\kappa\gamma - 1}] + \mathcal{O}\ab(\frac{1}{\alpha})& \hat{\bar{q}} &= \frac{1+\sigma^2}{\alpha}\ab(\gamma - \frac{1}{\kappa}) + \mathcal{O}\ab(\frac{1}{\alpha^2})\\
    q_0 &= \begin{cases}
      \frac{\kappa}{\rho}+\frac{\kappa(\rho-\kappa)}{\rho}\frac{1}{\kappa\gamma-1} + \frac{\rho-\kappa}{\rho(\kappa\gamma-1)}\frac{\sigma^2}{1+\sigma^2}\alpha + \mathcal{O}\ab(\frac{1}{\alpha}) & \text{for } \rho > \kappa \\
      1+\frac{\sigma^2}{\kappa\gamma-1}\frac{\kappa \rho}{\kappa - \rho} + \mathcal{O}\ab(\frac{1}{\alpha})& \text{for } \rho < \kappa 
    \end{cases} \\
    \chi &=\frac{\min\ab(\kappa, \rho)}{\kappa\gamma-1} + \mathcal{O}\ab(\frac{1}{\alpha})& \\
    \hat{\chi} &=  \ab(\gamma - \frac{1}{\kappa})\ab(\sigma^2  + \min\ab(\kappa, \rho)\frac{1+\sigma^2}{\alpha}) + \mathcal{O}\ab(\frac{1}{\alpha^2})\\
    \bar{\chi} &=\frac{1-\min\ab(\kappa,\rho)}{\kappa\gamma-1}\frac{1}{1+\sigma^2}\alpha   + \mathcal{O}\ab(1)& \\
    \hat{\bar{\chi}} &= \frac{1+\sigma^2}{\alpha} \ab(\gamma - \frac{1}{\kappa}) \ab(\sigma^2  + \min\ab(\kappa, \rho)\frac{1+\sigma^2}{\alpha}) + \mathcal{O}\ab(\frac{1}{\alpha^3}) 
  \end{align}
\end{proposition}

The key findings in the main paper (Results~\ref{result:decomposition}, \ref{result:asump_decomp}, \ref{result:implicit_regularization}, and \ref{result:asymptotic_coefficients}) are recovered from Result~\ref{result:complete_statement_W} and Proposition~\ref{prop:asympt}. Similarly, Result~\ref{result:gen_error_simple} is recovered from Result~\ref{result:complete_statement_E} and Proposition~\ref{prop:asympt}. In the following, we provide the detailed replica calculations to sequentially derive these foundational results: Results~\ref{result:complete_statement_W}, \ref{result:parameter_derivation}, \ref{result:complete_statement_E}, and Proposition~\ref{prop:asympt}.

\subsection{Replica System and Replicated Partition Function}
\label{sec:replica_system_and_replicated_partition_function}

In this subsection, we outline the replica formalism for evaluating the moments of the solution $W^*$, which forms the basis for the statistical characterization of the estimator in the high-dimensional limit.
To derive Result~\ref{result:complete_statement_W}, we first rely on two fundamental assumptions. The first assumption addresses the well-posedness of the statistical problem itself.
\begin{assumption}[Identifiability from Moments of a Scalar Statistic]\label{ass:identifiability_from_moments}
  Let $X$ be a random matrix taking values in $\mathbb{R}^{D \times D}$, and let $g: \mathbb{R}^{D \times D} \to \mathbb{R}$ be an arbitrary measurable function. 
  We assume that the probability law of $X$ is uniquely determined by the sequence of moments of the scalar statistic $g(X)$.
  More formally, for any two random matrices $X$ and $X'$, the condition
  \begin{equation}
  \mathbb{E}[(g(X))^n] = \mathbb{E}[(g(X'))^n] \quad \text{for all } n \in \mathbb{N}
  \end{equation}
  implies that $X$ and $X'$ are identically distributed.
  \end{assumption}
This is a technical assumption, positing that the random variable $X$ does not exhibit pathological behavior (such as that of a log-normal distribution) where its moments fail to uniquely define its distribution. Assuming that $W^*$ satisfies this property for a fixed task set $\mathcal{W}_0$, 
our goal is to compute its integer moments, $\E[g\ab(W^*)^p]$ for $p \in \mathbb{N}$.

Our starting point is to re-interpret the solution of the optimization problem, $W^*$, from a statistical mechanics perspective. For a fixed set of training instances $\mathcal{D} = \{ H^\mu \mid 1 \leq \mu \leq M \}$, we treat $W^*$ as a random variable drawn from a Gibbs-Boltzmann distribution, where the loss function $\mathcal{L}(W)$ acts as the energy function. In the zero-temperature limit ($\beta \to \infty$), this distribution concentrates on the global minimum of the loss:
\begin{equation}
  W^* = \argmin_{W}  \mathcal{L}(W)  \sim p(W \mid \mathcal{D}) =  \lim_{\beta \to \infty} \frac{\exp\ab(-\beta \mathcal{L}(W \mid \mathcal{D}))}{\int \odif{W} \exp\ab(-\beta \mathcal{L}(W \mid \mathcal{D}))} . \label{eq:boltzmann_distribution}
\end{equation}
Here, $\mathcal{L}(W \mid \mathcal{D})$ is the loss function for the model $W$ on the fixed learning instances $\mathcal{D}$, and $\odif{W} = \prod_{ij} \odif{W_{ij}}$. The normalization factor of Eq.~\eqref{eq:boltzmann_distribution} is the \textbf{partition function} $Z$ of the system:
\begin{equation}
  Z =  \int \odif{W} \exp\ab(-\beta \mathcal{L}(W \mid \mathcal{D})) .
\end{equation}
Using this notation, the $p$-th moment of $W^*$, averaged over the training data, is expressed as:
\begin{align}
  \E_{\mathcal{D}}[g\ab(W^*)^p] &= \E_{\mathcal{D}} \lim_{\beta \to \infty} \ab( \frac{1}{Z} \int \odif{W} \, g(W) \exp\ab(-\beta \mathcal{L}(W \mid \mathcal{D})) )^p \\
  &= \E_{\mathcal{D}} \lim_{\beta \to \infty} \lim_{n \to 0} Z^{n-p}  \ab( \int \odif{W} \, g(W) \exp\ab(-\beta \mathcal{L}(W \mid \mathcal{D})) )^p  .\label{eq:E_W_p_expression}
\end{align}
However, the expression in Eq.~\eqref{eq:E_W_p_expression} is difficult to average over the data distribution $\mathcal{D}$ because the data-dependent partition function $Z$ appears in the denominator.

To circumvent this difficulty, we employ the replica trick. The key insight is to represent the problematic term $Z^{n'-p}$ as an integral over $n'-p$ independent copies (or `replicas') of the system, which is valid for any integer $n' > p$:
\begin{align}
  Z^{n'-p} = \int \odif{W}_{n'-p} \,  \exp\ab(-\beta \sum_{a=1}^{n'-p} \mathcal{L}(W^a \mid \mathcal{D})) . \label{eq:Z_n_p_relation}
\end{align}
Here, $\odif{W_s} = \odif{W^1} \odif{W^2} \cdots \odif{W^{s}}$. This identity forms the basis for our second key assumption, which involves analytically continuing this expression from the domain of integers $n'$ to the limit $n' \to 0$.

\begin{assumption}[Analytic Continuation of the Partition Function]\label{ass:analytic_continuation_of_partition_function}  
  The identity in Eq.~\eqref{eq:Z_n_p_relation}, which is guaranteed to hold for integers $n' > p$, is assumed to remain valid in the limit $n' \to 0$. That is, for any integer $p$,
  \begin{equation}
     \lim_{n' \to 0} Z^{n'-p} = \lim_{n' \to 0} \int \odif{W}_{n'-p} \,  \exp\ab(-\beta \sum_{a=1}^{n'-p} \mathcal{L}(W^a \mid \mathcal{D})) . \label{eq:Z_n_p_relation_2}
  \end{equation}
\end{assumption}
While not mathematically rigorous for all systems, this assumption is standard practice in the replica method. With this assumption, Eq.~\eqref{eq:E_W_p_expression} can be rewritten by combining all terms into a single expectation over $n$ replicas:
\begin{align}
  \E_{\mathcal{D}}[g\ab(W^*)^p] &= \E_{\mathcal{D}} \lim_{\beta \to \infty} \lim_{n \to 0} Z^{n-p}  \ab( \int \odif{W} \, g(W) \exp\ab(-\beta \mathcal{L}(W \mid \mathcal{D})) )^p \\
  &= \E_{\mathcal{D}} \lim_{\beta \to \infty} \lim_{n \to 0}  Z^{n-p} \int \odif{W}_{p} \,  \ab(\prod_{a=1}^{p} g(W^a))   \exp\ab(-\beta \sum_{a=1}^{p} \mathcal{L}(W^a \mid \mathcal{D})) \\
  &=\lim_{\beta \to \infty} \lim_{n \to 0} \E_{\mathcal{D}}  \int \odif{W}_{n} \, \ab(\prod_{a=1}^{p} g(W^a))  \exp\ab(-\beta \sum_{a=1}^{n} \mathcal{L}(W^a \mid \mathcal{D})) \\ 
  &=  \lim_{n \to 0} \frac{\lim_{\beta \to \infty} \E_{\mathcal{D}}  \int \odif{W}_{n} \, \ab(\prod_{a=1}^{p} g(W^a))  \exp\ab(-\beta \sum_{a=1}^{n} \mathcal{L}(W^a \mid \mathcal{D})) }{\lim_{\beta \to \infty} \E_{\mathcal{D}} \ab[Z^n]} \\
  &=\lim_{n \to 0} \frac{ \lim_{\beta \to \infty} \E_{\mathcal{D}}  \int \odif{W}_{n} \, \ab(\prod_{a=1}^{p} g(W^a))  \exp\ab(-\beta \sum_{a=1}^{n} \mathcal{L}(W^a \mid \mathcal{D})) }{\lim_{\beta \to \infty} \E_{\mathcal{D}}  \int \odif{W}_{n} \,   \exp\ab(-\beta \sum_{a=1}^{n} \mathcal{L}(W^a \mid \mathcal{D}))} .\label{eq:E_W_p_expression_2}
\end{align}
The final expression in Eq.~\eqref{eq:E_W_p_expression_2} is the central result of this formalism. 
It shows that the desired moment, $\E_{\mathcal{D}}[g\ab(W^*)^p]$, can be interpreted as a $p$-point correlation function of the replicated variables $\{W^1, \dots, W^n\}$ drawn from an effective probability distribution:
\begin{align}
  p(W^1, \dots, W^n) \propto \lim_{\beta \to \infty} \E_{\mathcal{D}} \exp\ab(-\beta \sum_{a=1}^{n} \mathcal{L}(W^a \mid \mathcal{D})) ,\label{eq:p_W_n} 
\end{align}
in the $n\to 0$ limit.
The main advantage of this approach is that the difficult average over the data distribution $\mathcal{D}$ has been absorbed into the structure of a new joint distribution over $n$ replicated systems. We call this new system the \textbf{replica system}, and its partition function, $\E_{\mathcal{D}} \ab[Z^n]$, the data-averaged \textbf{replicated partition function}.

In the subsequent sections, we will analyze this replicated system. By calculating the replicated partition function, we will derive the correlations between replicas, which in turn will allow us to derive Result~\ref{result:complete_statement_W}.

\subsection{Analysis of the Data-Averaged Replicated Partition Function}

In this subsection, we analyze the data-averaged replicated partition function to derive the statistical properties of the optimal solution. Our immediate goal is to compute the partition function of the probability distribution Eq.~\eqref{eq:p_W_n}, which is the data (disorder) average of the replicated partition function:
\begin{equation}
  \E [Z^n] =  \E \int \odif{W}_n \, \exp \ab[  - \frac{\beta}{2} \sum_{\mu=1}^M\sum_{a=1}^n \ab( \sum_{ij} \ab( \frac{w_i^\mu w_j^\mu}{D} - W_{ij}^a ) H_{ij}^\mu + \epsilon^\mu)^2 - \frac{D\lambda\beta}{2} \sum_{ij}\sum_a \ab(W_{ij}^a)^2 ] .
\end{equation}

The overall disorder average $\E[\cdots]$ here is performed over three hierarchical stages of data generation. Each stage introduces a different source of randomness:
\begin{itemize}
    \item \textbf{$\E_{\mathcal{W}_0}$}: This denotes the average over the generation of the base pool of tasks, $\mathcal{W}_0 = \{\vb{w}^1, \dots, \vb{w}^{M_0}\}$. This expectation accounts for the randomness in the underlying low-dimensional structure from which all tasks are derived (i.e., the shared matrix $A$ and latent vectors $\{\vb{v}^\mu\}$).
    \item \textbf{$\E_{\mathcal{W} \mid \mathcal{W}_0}$}: This denotes the average over the sampling of the final training set $\mathcal{W}$. Specifically, it is the average over the process of drawing $M$ tasks uniformly and with replacement from the fixed base pool $\mathcal{W}_0$.
    \item \textbf{$\E_{H^\mu \mid \vb{w}^\mu}$}: This denotes the average over the generation of the data within a single training instance for a given task $\vb{w}^\mu$. This ``in-context" randomness comes from sampling the $L+1$ input vectors $\{\vb{x}_l^\mu\}$ and the corresponding label noise $\{\epsilon_l^\mu\}$.
\end{itemize}

Using the above notation, we can isolate the average over the data distribution ($H^\mu$ and noise) for a fixed task $\vb{w}^\mu$ by defining the function $\phi(\vb{W}, \vb{w}^\mu)$:
\begin{align}
 \phi(\vb{W}, \vb{w}^\mu) = \E_{H^\mu \mid \vb{w}^\mu} \exp \ab[-\frac{\beta}{2} \sum_a \ab(\sum_{ij} \ab(\frac{w^\mu_i w^\mu_j}{D} - W_{ij}^a) H^\mu_{ij} + \epsilon^\mu)^2 ].
\end{align}
With this definition, the replicated partition function becomes:
\begin{align}
 \E [Z^n] &= \int \odif{W}_n \,\exp\ab(- \frac{D\lambda\beta}{2} \sum_a \sum_{ij} \ab(W_{ij}^a)^2) \, \E_{\mathcal{W}_0} \E_{\{\vb{w}^\mu\}\mid {\cal W}_0} \prod_{\mu =1}^M \phi(\vb{W}, \vb{w}^\mu) \\
 &= \int \odif{W}_n \,\exp\ab(- \frac{D\lambda\beta}{2} \sum_a \sum_{ij} \ab(W_{ij}^a)^2)  \, \prod_{\mu=1}^{M_0} \E_{\mathcal{W}_0} \ab[\phi (\vb{W}, \vb{w}^\mu)^{\frac{M}{M_0}} ]  \\
 &= \E_{\mathcal{W}_0} \int \odif{W}_n\,\exp\ab(- \frac{D\lambda\beta}{2} \sum_a \sum_{ij} \ab(W_{ij}^a)^2)  \, \exp \ab[\frac{M}{M_0}  \sum_{\mu=1}^{M_0}  \log \phi (\vb{W}, \vb{w}^\mu)] .
\end{align}
In the last step, we assume self-averaging with respect to the tasks $\vb{w}^\mu$, allowing us to replace the sum over logarithms with its average value, which simplifies the expression to depend on the average of $\log\phi$:
\begin{align}
 \E [Z^n] &= \E_{\mathcal{W}_0} \int \odif{W}_n\,\exp\ab(- \frac{D\lambda\beta}{2} \sum_a \sum_{ij} \ab(W_{ij}^a)^2)  \, \exp \ab[M   \log \phi (\vb{W}, \bar{Q}, m)] .
\end{align}

To proceed, we introduce the following order parameters, which describe the macroscopic state of the system. 
We define the following order parameters:
\begin{align}
 Q^{ab} &= \frac{1}{M_0} \sum_{\mu=1}^{M_0}\frac{1}{D} ({W}^a\vb{w}^\mu)^\top ({W}^b \vb{w}^\mu) \\
 \bar{Q}^{ab} &= \frac{1}{D} \tr \left(({W}^a)^\top {W}^b\right) \\
 m^a &= \frac{1}{M_0} \sum_{\mu=1}^{M_0}\frac{1}{D} {\vb{w}^\mu}^\top {W}^a \vb{w}^\mu,
\end{align}
where the index $a,b \in \{1, ..., n\}$ denotes the replica index.
$Q^{ab}$ measures the task-conditioned overlap between replicas, $\bar{Q}^{ab}$ measures the direct overlap of the weight matrices, and $m^a$ measures the alignment of a replica's weights to the task structure.

We now re-evaluate the data-averaged term $\phi$ by performing the Gaussian integral. We define the vectors
\begin{align}
 \vb{v}^a = \pmqty{\mathsf{vec}\ab(\frac{\vb{w}^\mu {\vb{w}^\mu}^\top}{D} - W^a) \\ 1}, \qquad \vb{h} = \pmqty{\mathsf{vec}(H^\mu) \\ \epsilon^\mu},
\end{align}
which allows us to write the exponent as a quadratic form. The average over the zero-mean Gaussian vector $\vb{h}$ yields:
\begin{align} 
 \phi(\vb{W}, \vb{w}^\mu) &= \E_{H^\mu \mid \vb{w}^\mu} \exp\ab[-\frac{\beta}{2}  \sum_a \ab(\sum_{ij} \ab(\frac{w^\mu_i w^\mu_j}{D} - W_{ij}^a) H_{ij}^\mu + \epsilon^\mu)^2 ] \\
 &=  \E_{H^\mu \mid \vb{w}^\mu} \exp\ab(-\frac{\beta}{2}  \sum_a \ab(\vb{v}^a\cdot \vb{h})^2 ) \\
 &= \ab[\det \ab(I_{D^2 + 1} + 2\beta C^H V V^\top)]^{-\frac{1}{2}} \\
 &= \ab[\det \ab(I_{n} + 2\beta V^\top C^H V)]^{-\frac{1}{2}} ,
\end{align}
where $V \in \R^{(D^2 + 1) \times n}$ is the matrix whose columns are $\vb{v}^a$, and $C^H$ is the covariance matrix of $\vb{h}$. The entries of the matrix $V^\top C^H V$ can be computed as:
\begin{align}
 \ab(V^\top C^H V)_{ab} &= \sum_{(ij), (kl)}^{D^2}  V_{(ij),a} V_{(kl),b} C^H_{(ij)(kl)} + \sigma^2  \\
 &= \sum_{(ij), (kl)}^{D^2}  \ab(\frac{w^\mu_i w^\mu_j}{D} - W_{ij}^a) \ab(\frac{w^\mu_k w^\mu_l}{D} - W_{kl}^b)  \frac{\delta_{ik}}{D} \ab(\frac{D}{L}\delta_{jl}(1 + \sigma^2) + w_j w_l ) + \sigma^2 \\
 &= \frac{D}{L}(1+\sigma^2) \bar{Q}^{ab} + \ab(Q^{ab} - m^a - m^b + 1 ) + \sigma^2 \label{eq:C_H} .
\end{align}
Defining $\mathcal{M} \in \R^{n \times n}$ as the matrix whose elements are given by the expression above, we finally obtain:
\begin{align}
 \phi(\vb{W}, \vb{w}^\mu) &= \ab[\det \ab(I_n + \beta \mathcal{M})]^{-\frac{1}{2}} \\
 &= \E_{\vb{u}} \exp\ab[-\frac{\beta}{2}\sum_a \ab(u^a)^2]\\
 &= \exp\ab[\frac{1}{2}\mathcal{M}^{ab}\pdv{}{h^a h^b}]\exp\ab(-\frac{\beta}{2}\sum_a \ab(h^a)^2) , 
\end{align}
where $\vb{u} = (u^1, \cdots, u^n)^\top \in \R^n$ is a Gaussian random vector with zero mean and covariance given by
\begin{align}
 \mathcal{M}^{ab} &= \frac{D}{L}(1+\sigma^2)\bar{Q}^{ab} + \ab(Q^{ab} - m^a - m^b + 1 ) + \sigma^2 . \label{eq:def_of_mathcal_M}
\end{align}

We enforce the definitions of the order parameters by introducing their integral representations using the Dirac delta function:
\begin{align}
 1 &= \int \odif{Q}^{ab} \, \delta \ab(DM_0 Q^{ab} -  \sum_{\mu=1}^{M_0}  \ab(W^a \vb{w}^\mu)^\top \ab(W^b \vb{w}^\mu)) \\
 &= \int \odif{Q}^{ab} \odif{\hat{Q}^{ab}} \, \exp\ab[ - \frac{\hat{Q}^{ab}}{2} \ab(DM_0 Q^{ab} -  \sum_{\mu=1}^{M_0}  \ab(W^a \vb{w}^\mu)^\top \ab(W^b \vb{w}^\mu)) ] \\
 1 &= \int \odif{\bar{Q}^{ab}} \, \delta \ab(DM_0 \bar{Q}^{ab} - M_0 \tr\ab( \ab(W^a)^\top W^b)) \\
 &= \int \odif{\hat{\bar{Q}}^{ab}} \odif{\hat{\bar{Q}}^{ab}} \, \exp\ab[ - \frac{\hat{\bar{Q}}^{ab}}{2} \ab(DM_0 \bar{Q}^{ab} - M_0 \tr\ab( \ab(W^a)^\top W^b)) ] \\
 1 &= \int \odif{m^a} \, \delta \ab(DM_0 m^a - \sum_{\mu=1}^{M_0} \ab(\vb{w}^\mu)^\top W^a \vb{w}^\mu ) \\
 &= \int \odif{\hat{m}^a} \odif{m^a} \, \exp\ab[ - \frac{\hat{m}^a}{2} \ab(DM_0 m^a - \sum_{\mu=1}^{M_0} \ab(\vb{w}^\mu)^\top W^a \vb{w}^\mu ) ] ,
\end{align}
where the conjugate parameters $\{\hat{Q}^{ab}, \hat{\bar{Q}}^{ab}\}_{ab}, \{\hat{m}^a\}_{a}$ are introduced\footnote{The integration over the conjugate parameters is performed along the imaginary axis to ensure the positivity of the parameters.}. 
This allows us to express the averaged replicated partition function as a product of three terms:
\begin{align}
 \E_{} [Z^n] &=\int \prod_{ab} \ab(\odif{Q}^{ab} \odif{\hat{Q}^{ab}} \odif{\bar{Q}^{ab}} \odif{\hat{\bar{Q}}^{ab}}) \prod_{a} \ab(\odif{m^a} \odif{\hat{m}^a})  \, G_I G_E G_S , \label{eq:G_I_G_E_G_S}
\end{align}
where
\begin{align}
 G_I &= \exp\ab(-\frac{DM_0}{2}\sum_{ab}Q^{ab}\hat{Q}^{ab} - \frac{D M_0}{2}\sum_{ab} \bar{Q}^{ab}\hat{\bar{Q}}^{ab} - DM_0\sum_{a} m^a \hat{m}^a) \label{eq:def_G_I}\\
 G_S &= \E_{{\cal W}_0} \int \odif{W^n} \exp\left[ \frac{1}{2} \sum_\mu \sum_{ab} \hat{Q}^{ab} \left( (W^a \vb{w}^\mu)^\top (W^b \vb{w}^\mu) \right) + \frac{M_0}{2} \sum_{ab} \hat{\bar{Q}}^{ab} \tr\left( {W^a}^\top W^b \right) \right. \notag \\
 &\qquad \left. + \sum_\mu \sum_a \hat{m}^a {\vb{w}^\mu}^\top W^a \vb{w}^\mu - \frac{\lambda\beta M_0}{2} \sum_{a} \tr\left( {W^a}^\top W^a \right) \right] \label{eq:def_G_S}\\
 G_E &=  \exp\ab[M  \log  \phi(Q, \bar{Q}, m) ] \label{eq:def_G_E}.
\end{align}

\subsection{Replica-Symmetric (RS) Ansatz}

In this subsection, we now apply the replica-symmetric (RS) ansatz, which is a foundational step in simplifying the replicated system. 
The core idea is to assume that all $n$ replicas are statistically equivalent and interchangeable. 
It is important to note that the validity of the RS solution is not mathematically guaranteed for all complex systems; its correctness has only been rigorously proven for a limited class of models. 
However, particularly for convex optimization problems such as the one considered in this work, the replica method under the Assumption.~\ref{ass:analytic_continuation_of_partition_function}, \ref{ass:identifiability_from_moments}, and \ref{ass:replica_symmetry} has an extensive and successful track record. To date, there are no known instances where the method, when applied to such problems under standard supporting assumptions, has led to physically inconsistent or contradictory results. 
Therefore, it serves as a crucial and often surprisingly accurate starting point in the analysis of disordered systems.

\begin{assumption}[Replica Symmetry]\label{ass:replica_symmetry}
The RS ansatz parameterizes the order parameter matrices, which depend on pairs of replica indices $(a, b)$, in terms of a small number of macroscopic variables. This symmetry is expressed as:
\begin{align}
 Q^{ab} &= \frac{\chi}{\beta} \delta_{ab} + q,  \\
 \bar{Q}^{ab} &= \frac{\bar{\chi}}{\beta} \delta_{ab} + \bar{q}, \\
 \hat{Q}^{ab} &= -\beta \hat{q} \delta_{ab} + \beta^2 \hat{\chi}, \\
 \hat{\bar{Q}}^{ab} &= - \beta \hat{\bar{q}} \delta_{ab} + \beta^2 \hat{\bar{\chi}} \\
 m^a &= m \\
 \hat{m}^a &= \beta \hat{m}.
\end{align}
\end{assumption}

This assumption is essential to the replica method for two primary reasons. 
First, it drastically simplifies the problem by positing that the complex interactions between $n(n-1)/2$ pairs of replicas can be described by a small, fixed number of macroscopic order parameters. 
This reduces a problem with a large number of degrees of freedom to one of solving a few self-consistent equations. Second, the symmetric structure imposed by the ansatz is a necessary technical step to analytically continue the expressions from integer $n$ to the $n \to 0$ limit, which is the core of the replica trick. 
Without this simplification, the combinatorial structure of the replica indices would prevent a well-defined continuation.

Upon substituting the RS ansatz into the expression for the replicated partition function, the average over the data distribution $\E\left[Z^n\right]$ (Eq.~\eqref{eq:G_I_G_E_G_S}) can be expressed as an integral over the relevant order parameters. 
Let the set of these parameters be denoted by $\vb{x} = \{q, \bar{q}, m, \chi, \bar{\chi}, \hat{q}, \hat{\bar{q}}, \hat{m}\}$. The resulting expression takes the form:
\begin{align}
  \E\left[Z^n\right] = \int \left(\prod_{i \in \vb{x}} \mathrm{d}x_i\right) \exp\left[-\beta D M_0 f(\vb{x}) + \mathcal{O}(1) \right]
  \label{eq:replicated_Z_integral}
\end{align}
where $f(\vb{x})$ is a function of the order parameters.

In the asymptotic limit where $D \to \infty$, this integral can be evaluated using the saddle-point method (also known as the method of steepest descent). 
Consequently, the values of the order parameters are determined by finding the specific configuration $\vb{x}^*$ that extremizes the function $f(\vb{x})$. 
This procedure yields the saddle-point equations, which require the partial derivative of $f$ with respect to each order parameter $x_i$ to be zero:
\begin{equation}
    \frac{\partial f}{\partial x_i} = 0, \quad \text{for all } x_i \in \vb{x}. \label{eq:saddle_point_equations}
\end{equation}
The solution to this system of equations determines the macroscopic state of the system, providing the values of the order parameters that characterize its typical behavior in the asymptotic limit.

\subsection{Statistics of the solution $W^*$ (derivation of the Result.~\ref{result:complete_statement_W})}
In this subsection, we derive an equivalent and interpretable expression that characterizes the statistics of $W^*$.
Substituting the RS ansatz into the expression for $G_S$ (Eq.~\eqref{eq:def_G_S}) gives:
\begin{align}
 G_S &= \E_{{\cal W}_0} \int \odif{W_n} \exp\left[ \frac{1}{2} \sum_\mu \sum_{ab} \hat{Q}^{ab} \left( (W^a \vb{w}^\mu)^\top (W^b \vb{w}^\mu) \right) + \frac{M_0}{2} \sum_{ab} \hat{\bar{Q}}^{ab} \tr\left( {W^a}^\top W^b \right) \right. \notag \\
 &\qquad \left. + \sum_\mu \sum_a \hat{m}^a {\vb{w}^\mu}^\top W^a \vb{w}^\mu - \frac{\lambda\beta M_0}{2} \sum_{a} \tr\left( {W^a}^\top W^a \right) \right] \\
 &= \E_{{\cal W}_0} \int \odif{W_n} \exp\left[ -\frac{\beta\hat{q}}{2} \sum_\mu \sum_a  (W^a \vb{w}^\mu)^\top (W^a \vb{w}^\mu) + \frac{\beta^2\hat{\chi}}{2} \sum_\mu \sum_a  (W^a \vb{w}^\mu)^\top \sum_b (W^b \vb{w}^\mu) \right. \notag \\
 &\qquad - \frac{M_0\beta\hat{\bar{q}}}{2} \sum_a \tr\left( {W^a}^\top W^a \right)  + \frac{M_0 \beta^2 \hat{\bar{\chi}} }{2}   \tr\left( \sum_a {W^a}^\top \sum_b W^b \right)  \notag \\
 &\qquad \left. + \beta   \hat{m} \sum_\mu \sum_a {\vb{w}^\mu}^\top W^a \vb{w}^\mu - \frac{\lambda\beta M_0}{2} \sum_{a} \tr\left( {W^a}^\top W^a \right) \right] .
\end{align}
This allows us to decouple the replicas through the following Hubbard-Stratonovich transformation:
\begin{align}
  \exp\ab[\frac{\beta^2\hat{\chi}}{2} \sum_\mu \sum_a  (W^a \vb{w}^\mu)^\top \sum_b (W^b \vb{w}^\mu)] &= \int \Odif{\vb{\xi}_{M_0}} \exp\ab[\beta \sqrt{\hat{\chi}} \sum_\mu \sum_a  (W^a \vb{w}^\mu)^\top \vb{\xi}^\mu] \\
  \exp\ab[\frac{M_0\beta^2\hat{\bar{\chi}}}{2} \tr \ab(\sum_a {W^a}^\top \sum_b W^b)] &= \int \Odif{\Lambda} \exp\ab[\beta \sqrt{M_0\hat{\bar{\chi}}}  \tr \ab({W^a}^\top \Lambda)] ,
\end{align}
where $\vb{\xi} = (\vb{\xi}^\mu \in \mathbb{R}^D)_{1 \le \mu \le M_0}$, $\Lambda \in \mathbb{R}^{D \times D}$, and $\Odif{\vb{\xi}_{M_0}} $ and $\Odif{\Lambda}$ are the standard Gaussian measures defined as 
$\Odif{\vb{\xi}_{M_0}} = \prod_{\mu=1}^{M_0} \exp\ab(-{\vb{\xi}^\mu}^\top \vb{\xi}^\mu/2) / \sqrt{2\pi}^D \odif{\vb{\xi}^\mu}$, $\Odif{\Lambda} = \prod_{ij} \exp\ab(-\Lambda_{ij}^2/2) / \sqrt{2\pi} \odif{\Lambda_{ij}}$.
Using these transformations, we can rewrite the expression for $G_S$ as:
\begin{align}
 G_S &= \E_{{\cal W}_0} \int \odif{W_n} \Odif{\vb{\xi}} \Odif{\Lambda} \, \exp\left[ -\beta \sum_a \left[\frac{1}{2}\tr\ab(W^a\ab(\hat{q}\sum_\mu \vb{w}^\mu {\vb{w}^\mu}^\top + (\lambda + \hat{\bar{q}} )M_0 I_D){W^a}^\top)  \right. \right. \notag \\
 &\qquad \left. \left. -\sum_\mu \ab(W^a\vb{w}^\mu)^\top \ab(\sqrt{\hat{\chi}}\vb{\xi}^\mu + \hat{m}\vb{w}^\mu) - \sqrt{M_0}\sqrt{\hat{\bar{\chi}}}\tr \ab(W^a \Lambda)  \right] \right] \\
 &= \E_{{\cal W}_0} \int \odif{W_n} \Odif{\vb{\xi}} \Odif{\Lambda} \,  \exp\ab[-\frac{\beta}{2} \sum_a f(W^a) ] , \label{eq:G_S_expression}
\end{align}
where we have defined
\begin{align}
 f(W) = \tr\ab[W\ab(\hat{q} \sum_\mu\vb{w}^\mu{\vb{w}^\mu}^\top + (\lambda+\hat{\bar{q}}) M_0 I_D)W^\top] - 2\sum_\mu \ab(\sqrt{\hat{\chi}}\vb{\xi}^\mu+\hat{m}\vb{w}^\mu)^\top W\vb{w}^\mu - 2\sqrt{M_0}\sqrt{\hat{\bar{\chi}}} \tr\ab(W \Lambda^\top) .
\end{align}

By extremizing the exponent of the integrand with respect to the order parameters, we obtain the saddle-point equations:
\begin{align}
 q &= \frac{1}{M_0} \E_{\mathcal{W}_0} \E_{\vb{\xi}, \Lambda} \ab[ \frac{1}{D}  \sum_{\mu=1}^{M_0}\ab(\hat{W} \vb{w}^{\mu})^\top \ab(\hat{W} \vb{w}^{\mu}) ] \label{eq:saddle_point_q} \\
 \bar{q} &= \frac{1}{D} \E_{\mathcal{W}_0} \E_{\vb{\xi}, \Lambda} \tr \ab({\hat{W}}^\top \hat{W}) \label{eq:saddle_point_bar_q} \\
 m &= \frac{1}{M_0} \E_{\mathcal{W}_0} \E_{\vb{\xi}, \Lambda} \ab[ \frac{1}{D}  \sum_{\mu=1}^{M_0}{\vb{w}^{\mu}}^\top \hat{W} \vb{w}^{\mu} ] \label{eq:saddle_point_m} \\
 \chi &= \frac{1}{M_0} \E_{\mathcal{W}_0} \E_{\vb{\xi}, \Lambda} \ab[  \sum_{\mu=1}^{M_0}   \ab({\vb{w}^{\mu}}^\top \ab(\hat{q}\sum_{\nu=1}^{M_0} \vb{w}^{\nu}{\vb{w}^{\nu}}^\top + (\lambda+\hat{\bar{q}}) M_0 I_D)^{-1} \vb{w}^{\mu})] \label{eq:saddle_point_chi} \\
 \bar{\chi} &=  \E_{\mathcal{W}_0} \E_{\vb{\xi}, \Lambda} \tr \ab[\ab(\hat{q}\sum_{\mu=1}^{M_0} \vb{w}^{\mu}{\vb{w}^{\mu}}^\top + (\lambda+\hat{\bar{q}}) M_0 I_D)^{-1}] \label{eq:saddle_point_bar_chi}.
\end{align}
Here, the optimal weight matrix $\hat{W}$ that minimizes the effective problem $f(W)$ is given by:
\begin{align}
 \hat{W} &= \argmin_{W} f(W) \\
 &= \argmin_{W}\ab[ \frac{\lambda + \hat{\bar{q}}}{2}\tr\ab(WW^\top) +\frac{\hat{q}}{2} \tr \ab(WSW^\top) - \tr\ab(\ab(\sqrt{\hat{\chi}}T + \sqrt{\hat{\bar{\chi}}}R + \hat{m}S)W^\top) ] \label{eq:effective_problem}\\
 &=\ab(\sqrt{\hat{\chi}}T + \sqrt{\hat{\bar{\chi}}}R + \hat{m}S)  \ab(\hat{q}S + \ab(\lambda + \hat{\bar{q}})I_D)^{-1} \label{eq:effective_problem_solution}, 
\end{align}
where we have defined the following matrices:
\begin{align}
 T = \frac{1}{M_0}\sum_{\mu=1}^{M_0} \vb{\xi}^\mu{\vb{w}^\mu}^\top, \quad R = \ab(R_{ij})_{i,j=1}^D \sim \mathcal{N}\ab(0, \frac{1}{M_0}), \quad S = \frac{1}{M_0} \sum_{\mu=1}^{M_0} \vb{w}^\mu{\vb{w}^\mu}^\top, 
\end{align}

The crucial insight from Eq.~\eqref{eq:G_S_expression} is that the expression for $G_S$ completely decouples over the replica indices $a$. 
This implies that the joint probability distribution for the $n$-replica system factorizes into a product of identical and independent distributions for each replica $W^a$. 

Our original goal is to compute the moments of a statistic $g(W^*)$, where $W^*$ is the solution to the original optimization problem. 
Within the replica framework, as seen in Section~\ref{sec:replica_system_and_replicated_partition_function}, the $p$-th moment $\E_{\mathcal{D}}[g(W^*)^p]$ corresponds to the expectation of a $p$-body correlation, $\E[g(W^1)g(W^2)\cdots g(W^p)]$, in the $n \to 0$ limit. 
Due to the factorization, this simplifies significantly:
\begin{align}
  \E_{\mathcal{D}}\left[g(W^*)^p\right] &=   \E_{\{W^a\}} \left[g(W^1)g(W^2)\cdots g(W^p)\right] \\
  &= \E_{\mathcal{W}_0} \int \Odif{\vb{\xi}} \Odif{\Lambda} \, \left[\lim_{\beta \to \infty} \frac{\int \odif{W} \, g(W)\exp\left[-\frac{\beta}{2}  f(W) \right]}{\int \odif{W} \, \exp\left[-\frac{\beta}{2} f(W) \right]} \right]^p \\
  &= \E_{\mathcal{W}_0, \vb{\xi}, \Lambda} \left[ g(\hat{W})^p \right]. \label{eq:moment_decomposition}
\end{align}
We now invoke our technical assumption, Assumption~\ref{ass:identifiability_from_moments}, which posits that if the moments of a scalar statistic of two random matrices are identical, then the matrices themselves are identically distributed. This directly leads to our main result:
\begin{equation}
  W^* \overset{\mathsf{d}}{=} \hat{W}.
\end{equation}
This establishes that the solution to the original, complex optimization problem, $W^*$, is distributionally equivalent to $\hat{W}$, the solution of the much simpler effective quadratic problem defined by $f(W)$. 
This completes the derivation of Result~\ref{result:complete_statement_W}.

\subsubsection{Interpretation of the Auxiliary Fields}
It is instructive to provide a physical interpretation of the auxiliary fields, $\vb{\xi}$ and $\Lambda$ in Eq.~\eqref{eq:G_S_expression}. 
The Hubbard-Stratonovich transformation replaced the difficult average over the data distribution with an average over simpler, 
data-independent Gaussian fields $\vb{\xi}$ and $\Lambda$. 
Eq.~\eqref{eq:moment_decomposition} shows that the auxiliary field $\vb{\xi}$ and $\Lambda$ can be interpreted as an effective random field that embodies the statistical fluctuations of the training data. 
The solution $\hat{W}$ is then understood as the optimal response to a particular realization of these effective data fluctuations.

\subsection{Simplified Expressions for the Saddle-Point Equations (derivation of the Result.~\ref{result:parameter_derivation})}
In this subsection, we derive a set of equations that determine the order parameters and their conjugate variables. 
By simplifying the saddle-point equations under the RS assumption, we obtain expressions suitable for numerical computation.

\subsubsection{Order Parameters}\label{sec:order_parameters}

Next, we derive the simplified expressions for the saddle-point equations in Eqs.~\eqref{eq:saddle_point_q}-\eqref{eq:saddle_point_bar_chi}.
First, we define $\mathcal{M} = \hat{q}S + (\lambda+\hat{\bar{q}})I_D$.
Noting that $\mathcal{M}^{-1}$ is commuting with $S$, $T^\top T = D/ M_0 S$, $AA^\top S = S$, and Lemma~\ref{lemma:explicit_expressions_for_matrix_moments},
we have the following relations:

\begin{align}
  q &= \frac{1}{D}\E\ab[\tr\ab(S{W^*}^\top W^*)] \\
  &= \frac{1}{D} \E\ab[\tr\ab(\hat{\chi}S\mathcal{M}^{-1}T^\top T \mathcal{M}^{-1} + \hat{\bar{\chi}}S\mathcal{M}^{-1}R^\top R \mathcal{M}^{-1} + \hat{m}^2 S\mathcal{M}^{-1}S^2 \mathcal{M}^{-1})] \\
  &= \frac{1}{M_0} \ab[ \hat{\chi} \E\ab[\tr\ab(S^2\mathcal{M}^{-2})] +\hat{\bar{\chi}} \E\ab[\tr\ab(S\mathcal{M}^{-2})] + \kappa \hat{m}^2 \E\ab[\tr\ab(S^3 \mathcal{M}^{-2})] ] \\
  &= \frac{1}{\kappa\hat{q}^2}\ab[\hat{\chi}\ab(1+2zg_S(z)+z^2g_S'(z))  + \hat{\bar{\chi}}\ab(g_S(z) + zg_S'(z)) + \kappa\hat{m}^2 \ab(1+2z+3z^2g_S(z) + z^3g_S'(z)) ]  
\end{align}
\begin{align}
  \bar{q} &= \frac{1}{D}\E\ab[\tr\ab({W^*}^\top W^*)] \\
  &= \frac{1}{D} \E\ab[\tr\ab(\hat{\chi}\mathcal{M}^{-1}T^\top T \mathcal{M}^{-1} + \hat{\bar{\chi}}\mathcal{M}^{-1}R^\top R \mathcal{M}^{-1} + \hat{m}^2 \mathcal{M}^{-1}S^2 \mathcal{M}^{-1})] \\
  &= \frac{1}{M_0} \ab[ \hat{\chi} \E\ab[\tr\ab(S\mathcal{M}^{-2})] +\hat{\bar{\chi}} \E\ab[\tr\ab(\mathcal{M}^{-2})] + \kappa \hat{m}^2\E\ab[\tr\ab(S^2 \mathcal{M}^{-2})] ] \\
  &= \frac{1}{\kappa\hat{q}^2} \ab[ \hat{\chi} \ab(g_S(z) + z g'_S(z)) + \hat{\bar{\chi}} g'_S(z) +  \kappa\hat{m}^2 \ab(1+2zg_S(z) + z^2g'_S(z))] 
\end{align}

\begin{equation}
  m = \frac{1}{D}\E\ab[\tr\ab(S{W^*})] = \frac{\hat{m}}{D}\E\ab[\tr\ab(S^2\mathcal{M}^{-1})] = \frac{\hat{m}}{\hat{q}} \ab(1 + z + z^2g_S(z))
\end{equation}

\begin{equation}
  \bar{m} = \frac{1}{D}\E\ab[\tr\ab(W^*)] = \frac{\hat{m}}{D}\E\ab[\tr\ab(S\mathcal{M}^{-1})] = \frac{\hat{m}}{\hat{q}} \ab(1 + z g_S(z))
\end{equation}
\begin{equation}
  \chi = \frac{1}{M_0}\E\ab[\tr\ab(S\mathcal{M}^{-1})] = \frac{1}{\kappa \hat{q}} \ab(1 + zg_S(z)) 
\end{equation}
\begin{equation}
  \bar{\chi} = \frac{1}{M_0}\E\ab[\tr\ab(\mathcal{M}^{-1})] = \frac{1}{\kappa \hat{q}} g_S(z) .
\end{equation}

Here, we used Lemma~\ref{lemma:trace_identity} to simplify Eq.~\eqref{eq:saddle_point_q_0_before} to Eq.~\eqref{eq:saddle_point_q_0_after}.
We further introduce the order parameters $q_0$ and $m_0$, which characterize the correlation between the weight matrix and the structure $A$ in the context of the IDG error, as follows:

\begin{align}
  q_0 &= \frac{1}{r}\E\ab[\tr\ab(AA^\top {W^*}^\top W^*)] \\
  &= \frac{1}{r} \E\ab[\tr\ab(\hat{\chi}AA^\top \mathcal{M}^{-1}T^\top T \mathcal{M}^{-1}+ \hat{\bar{\chi}}AA^\top \mathcal{M}^{-1}R^\top R \mathcal{M}^{-1}+ \hat{m}^2 AA^\top \mathcal{M}^{-1}S^2 \mathcal{M}^{-1})] \\
  &= \frac{1}{r\kappa} \ab[ \hat{\chi} \E\ab[\tr\ab(S \mathcal{M}^{-2})] +\hat{\bar{\chi}} \E\ab[\tr\ab(A^\top A\mathcal{M}^{-2})] + \hat{m}^2\kappa \E\ab[\tr\ab(S^2 \mathcal{M}^{-2})] ] \label{eq:saddle_point_q_0_before} \\
  &= \frac{1}{\rho \kappa\hat{q}^2} \ab[\hat{\chi} \ab(g_S(z) + z g'_S(z)) +  \rho \hat{\bar{\chi}} g'_{\tilde{S}}(z) + \kappa\hat{m}^2\ab(1+2zg_S(z) + z^2g'_S(z))] \label{eq:saddle_point_q_0_after}
\end{align}
\begin{equation}
  m_0 = \frac{1}{r}\E\ab[\tr\ab(A^\top {W^*} A)] = \frac{\hat{m}}{r}\E\ab[\tr\ab(S\mathcal{M}^{-1})] = \frac{\hat{m}}{\rho \hat{q}} \ab(1 + z g_S(z))
\end{equation}

\subsubsection{Conjugate Parameters}\label{sec:conjugate_parameters}
The saddle point equations which define the conjugate parameters are followed by $G_I$ (Eq.~\eqref{eq:def_G_I}) and $G_S$ (Eq.~\eqref{eq:def_G_S}) in Eq.~\eqref{eq:G_I_G_E_G_S}.
In the RS ansatz, there is permutation symmetry among the replica indices.
As a result, the saddle point equations can be written for any $a\neq b$ as follows:
\begin{align}
  \hat{\chi} &= \gamma \int \Odif{\vb{z}}\Odif{z} \, u^a u^b \exp\ab[-\frac{\beta}{2}\sum_c \ab(u^c)^2] \label{eq:saddle_point_chi_beforeRS} \\
  \hat{\bar{\chi}} &= \gamma \frac{1}{\alpha}\ab(1+\sigma^2) \int \Odif{\vb{z}}\Odif{z} \, u^a u^b \exp\ab[-\frac{\beta}{2}\sum_c \ab(u^c)^2] \\
  \hat{q} &= \gamma \int \Odif{\vb{z}}\Odif{z} \, \ab(1+\beta u^a u^b -\beta \ab(u^a)^2) \exp\ab[-\frac{\beta}{2}\sum_c \ab(u^c)^2] \\
  \hat{\bar{q}} &= \gamma \frac{1}{\alpha}\ab(1+\sigma^2) \int \Odif{\vb{z}}\Odif{z} \, \ab(1+\beta u^a u^b -\beta \ab(u^a)^2) \exp\ab[-\frac{\beta}{2}\sum_c \ab(u^c)^2] \\
  \hat{m} &= \gamma \int \Odif{\vb{z}}\Odif{z} \, \ab(1-u^a \sum_{b=1}^n u^b) \exp\ab[-\frac{\beta}{2}\sum_c \ab(u^c)^2]  \label{eq:saddle_point_m_beforeRS},
\end{align}
where $u^a$ is a Gaussian variable with covariance matrix $\mathcal{M}^{ab}$ (see Eq.~\eqref{eq:def_of_mathcal_M}). Under the RS ansatz, 
$u^a$ can be represented using $n+1$ independent standard normal random variables $z$ and $\vb{z} = \{z^a\}_{a=1, \cdots, n}$ as:
\begin{align}
  u^a = z^a \sqrt{\frac{\frac{1}{\alpha}\ab(1+\sigma^2)\bar{\chi} +  \chi }{\beta}}  + z  \sqrt{\frac{1}{\alpha}\ab(1+\sigma^2) \bar{q}+ q-2m+1+\sigma^2} + h_{\text{source}} ,
\end{align}
where we introduced $h_{\text{source}}$ as a source term used to generate moments of $u^a$, which is evaluated at $h_{\text{source}}=0$.
Here, by changing variables as $\tilde{z}^a = \sqrt{z^a/\beta}$ and taking the limit $\beta\to\infty$, 
the integral of $\tilde{z}^a$ with Gaussian weight concentrates on $\tilde{z}^a = \ab(\tilde{z}^a)^*$ as:
\begin{align}
  \int \Odif{\vb{z}} \, \ab(\prod_a f\ab(u^a)) \exp\ab[-\frac{\beta}{2}\sum_a \ab(u^a)^2] &= \frac{\int \Odif{\vb{z}} \, \ab(\prod_a f\ab(u^a)) \exp\ab[-\frac{\beta}{2}\sum_a \ab(u^a)^2]}{\int \Odif{\vb{z}} \,  \exp\ab[-\frac{\beta}{2}\sum_a \ab(u^a)^2]} \quad (n\to 0) \\
  &=  \frac{\prod_a \int \odif{\tilde{z}^a} \, f\ab(u^a) \exp\ab[-\frac{\beta}{2} \ab(\ab(u^a)^2 + \ab(\tilde{z}^a)^2) ]}{\prod_a \int \odif{\tilde{z}^a} \,  \exp\ab[-\frac{\beta}{2} \ab(\ab(u^a)^2 + \ab(\tilde{z}^a)^2) ]} \\
  &= \left. \prod_a f(u^a) \right|_{\tilde{z}^a = \ab(\tilde{z}^a)^*} \quad (\beta \to \infty) ,
\end{align}
where $f(u^a)$ is a function of $u^a$, and the value of $\ab(\tilde{z}^a)^*$ is the solution of the following optimization problem:
\begin{align}
  \ab(\tilde{z}^a)^* &= \argmin_{\tilde{z}^a} \ab[\frac{1}{2} (\tilde{z}^a)^2 + \frac{1}{2} \ab(u^a)^2 ] \\
  &= \argmin_{\tilde{z}^a} \ab[\frac{1}{2} (\tilde{z}^a)^2 + \frac{1}{2} \ab( \sqrt{\frac{1}{\alpha}\ab(1+\sigma^2)\bar{\chi} +  \chi }\, \tilde{z}^a + \sqrt{\frac{1}{\alpha}\ab(1+\sigma^2) \bar{q}+ q-2m+1+\sigma^2} \, {z} + h_{\text{source}})^2 ] \\
  &= - \frac{\sqrt{\frac{1}{\alpha}\ab(1+\sigma^2)\bar{\chi} +  \chi }\ab( \sqrt{\frac{1}{\alpha}\ab(1+\sigma^2) \bar{q}+ q-2m+1+\sigma^2} \, {z} + h_{\text{source}})}{1+\chi+\frac{1}{\alpha}\ab(1+\sigma^2)\bar{\chi}} .
\end{align}
We denote $u^a$ evaluated at $\tilde{z}^a = (\tilde{z}^a)^*$ as $\tilde{u}^*$, i.e., 
\begin{align}
  \tilde{u}^* &=  \ab(\tilde{z}^a)^* \sqrt{\frac{1}{\alpha}\ab(1+\sigma^2)\bar{\chi} +  \chi} + z  \sqrt{\frac{1}{\alpha}\ab(1+\sigma^2) \bar{q}+ q-2m+1+\sigma^2} + h_{\text{source}} \\
  &= \frac{\sqrt{\frac{1}{\alpha}\ab(1+\sigma^2)\bar{q}+ q - 2m+ 1 + \sigma^2} \, z + h_{\text{source}}}{1+\chi + \frac{1}{\alpha}\ab(1+\sigma^2) \hat{\bar{\chi}}} ,
\end{align}
which does not depend on the replica index. 
By substituting this property and taking the limit $n \to 0$, the saddle point equations Eqs.~(\ref{eq:saddle_point_chi_beforeRS})--(\ref{eq:saddle_point_m_beforeRS}) can be rewritten as follows:
\begin{align}
  \hat{\chi} &= \left. \gamma \int \Odif{z} \, \ab(\tilde{u}^*)^2 \right|_{h_\text{source}=0} = \gamma \frac{\frac{1}{\alpha}\ab(1+\sigma^2)\bar{q} + q - 2m + 1 + \sigma^2}{\ab(1 + \chi + \frac{1}{\alpha} \ab(1+\sigma^2)\bar{\chi})^2} \\
  \hat{\bar{\chi}} &= \left.\frac{1}{\alpha}\ab(1+\sigma^2) \gamma \int \Odif{z} \, \ab(\tilde{u}^*)^2 \right|_{h_\text{source}=0} = \gamma \frac{1}{\alpha}\ab(1+\sigma^2) \frac{\frac{1}{\alpha}\ab(1+\sigma^2)\bar{q} + q - 2m + 1 + \sigma^2}{\ab(1 + \chi + \frac{1}{\alpha} \ab(1+\sigma^2)\bar{\chi})^2} \\
  \hat{q} &= \left.\gamma \int \Odif{z} \odv{\tilde{u}^*}{h_{\text{source}}} \right|_{h_\text{source}=0} = \gamma \frac{1}{1+\chi + \frac{1}{\alpha}\ab(1+\sigma^2)\bar{\chi}} \\
  \hat{\bar{q}} &=\left. \frac{1}{\alpha}\ab(1+\sigma^2) \gamma \int \Odif{z} \odv{\tilde{u}^*}{h_{\text{source}}} \right|_{h_\text{source}=0} = \frac{1}{\alpha}\ab(1+\sigma^2) \gamma \frac{1}{1+\chi + \frac{1}{\alpha}\ab(1+\sigma^2)\bar{\chi}} \\
  \hat{m} &= \left. \gamma \int \Odif{z} \odv{\tilde{u}^*}{h_{\text{source}}} \right|_{h_\text{source}=0}  = \gamma \frac{1}{1+\chi + \frac{1}{\alpha}\ab(1+\sigma^2)\bar{\chi}} .
\end{align}

These results in Section~\ref{sec:order_parameters} and Section~\ref{sec:conjugate_parameters} are summarized in Result~\ref{result:parameter_derivation}.

\subsection{Generalization Error (derivation of the Result~\ref{result:complete_statement_E})}
In this subsection, we compute the MSE in each protocol from Result.~\ref{result:complete_statement_W} using the order parameters defined by Result.~\ref{result:parameter_derivation}.
For the generalization error, we need to compute the correlation between the prediction using the optimal weight matrix $W^*$ and the ground-truth label $y$.

\subsubsection{Task Memorization}

The effective feature matrix of the prompt in the task memorization setting is given by:
\begin{equation}
  H = \frac{1}{M_0} \frac{D}{\tilde{L}} \sum_{\mu = 1}^{M_0} \ab[ \vb{x}_{L+1} {\vb{w}^\mu}^\top  \sum_{l=1}^{\tilde{L}} \vb{x}_l \vb{x}_l^\top],
\end{equation}
where $\vb{w}^* = 1/ M_0 \sum_{\mu=1}^{M_0} \vb{w}^\mu$ is the averaged learned task in the pre-training.
We have the following relation:
\begin{align}
  \mathcal{E}_{\text{TM}} &= \E \ab[y - \tr(W^* H^\top)]^2 \\
  &= 1-2\sum_{ij} \E\ab[y H_{ij} W^*_{ij}] + \sum_{ijkl} \E\ab[H_{ij} H_{kl} W^*_{ij} W^*_{kl}] \\
  &= 1 - 2\frac{1}{D}\frac{1}{M_0}\sum_\mu \tr\ab({\vb{w}^\mu}^\top W^* \vb{w}^\mu) + \frac{1}{\tilde{\alpha}} \frac{1}{D}\tr\ab({W^*}^\top W^*) + \frac{1}{D} \frac{1}{M_0}\sum_\mu \ab(W^* \vb{w}^\mu)^\top \ab(W^* \vb{w}^\mu) \\
  &= 1 - 2m + q + \frac{1}{\tilde{\alpha}} \bar{q} .
\end{align}

\subsubsection{In-Distribution Generalization}

The effective feature matrix of the prompt in the in-distribution setting is given by:
\begin{equation}
  H = \frac{D}{\tilde{L}} \ab[ \vb{x}_{L+1} \ab(A \vb{v}^*)^\top  \sum_{l=1}^{\tilde{L}} \vb{x}_l \vb{x}_l^\top],
\end{equation}
where $\vb{v}^* \sim \mathcal{N}(0, I_r) \in \R^r$.
We have the following relation:
\begin{align}
  \mathcal{E}_{\text{IDG}} &= \E \ab[y - \tr(W^* H^\top)]^2 \\
  &= 1-2\sum_{ij} \E\ab[y H_{ij}]\E\ab[W^*_{ij}] + \sum_{ijkl} \E\ab[H_{ij} H_{kl} W^*_{ij} W^*_{kl}] \\
  &= 1 - \frac{2}{r} \sum_{ijkl} \E\ab[A_{ik}A_{jl}  W^*_{ij}] \E\ab[v_kv_l] + \frac{1}{D} \sum_{ijl} \E\ab[\ab(\frac{D}{\tilde{L}}\delta_{ij}+ \sum_{st}A_{js}A_{lt}v_sv_t) W^*_{ij} W^*_{il}] \\
  &= 1 - 2\frac{1}{r}\E \ab[\tr \ab(A^\top W^* A)] + \frac{1}{r}\E \ab[\tr \ab(\ab(W^* A)^\top W^* A)] + \frac{1}{\tilde{\alpha}}\frac{1}{D}\tr\ab({W^*}^\top W^*) \\
  &= 1 - 2{m}_0 + q_0 +  \frac{1}{\tilde{\alpha}} \bar{q} .
\end{align}

\subsubsection{Out-of-Distribution Generalization}
The effective feature matrix of the prompt in the out-of-distribution setting is given by:
\begin{equation}
  H = \frac{D}{\tilde{L}} \ab[ \vb{x}_{L+1} {\vb{w}^*}^\top  \sum_{l=1}^{\tilde{L}} \vb{x}_l \vb{x}_l^\top],
\end{equation}
where $\vb{w}^* \sim \mathcal{N}(0, I_D) \in \R^D$.
We have the following relation:
\begin{align}
  \mathcal{E}_{\text{ODG}} &= \E \ab[y - \tr(W^* H^\top)]^2 \\
  &= 1-2\sum_{ij} \E\ab[y H_{ij}]\E\ab[W^*_{ij}] + \sum_{ijkl} \E\ab[H_{ij} H_{kl}]\ab[W^*_{ij} W^*_{kl}] \\
  &= 1 - \frac{2}{D} \sum_{ij} \E\ab[w_i w_j] \E\ab[W^*_{ij}] + \frac{1}{D} \sum_{ijl} \E\ab[\frac{D}{\tilde{L}}\delta_{ij}+ w_jw_l] \ab[W^*_{ij} W^*_{il}] \\
  &= 1 - 2\frac{1}{D}\tr W^* + \ab(1+ \frac{1}{\tilde{\alpha}}) \frac{1}{D}\tr\ab({W^*}^\top W^*) \\
  &= 1 - 2\bar{m} + \ab(1+ \frac{1}{\tilde{\alpha}}) \bar{q} .
\end{align}

These results are summarized in Result~\ref{result:complete_statement_E}.

\subsection{Asymptotics under $\alpha\gg1$ (derivation of Proposition~\ref{prop:asympt})}

The proof proceeds by performing an asymptotic expansion of the saddle-point equations, presented in Result~\ref{result:parameter_derivation}, in the limit where $\alpha \to \infty$.

The key insight is that in this limit, the system of equations simplifies considerably. First, we observe that the conjugate parameters $\hat{\bar{q}}$ and $\hat{\bar{\chi}}$ vanish as $\mathcal{O}(1/\alpha)$ or faster. This causes the parameter $z = -(\lambda + \hat{\bar{q}})/\hat{q}$ to converge to a finite constant that is dependent of $\alpha$. 

By substituting these simplified forms back into the equations for the primary order parameters ($q, m, \bar{q}$, etc.) and retaining only the leading-order terms in $\alpha$, the coupled system of equations becomes analytically solvable. This systematic expansion directly yields the explicit expressions for each parameter as stated in the proposition.

\section{Proof of Theorem~\ref{theorem:eigenvalue_distribution_of_S}}
\label{app:eigen_proof}

In this section, we proof Theorem~\ref{theorem:eigenvalue_distribution_of_S}.

From Lemma~\ref{lemma:resolvent_S}, the resolvent of $S$ is given by:
\begin{equation}
    g_S(z) \;=\; - (1-\rho)\frac{1}{z} + \rho \frac{-(\kappa \rho z + \rho - \kappa) - \sqrt{\left(\kappa \rho z + \rho - \kappa \right)^2 - 4\rho^2 \kappa z}}{2\rho z}.
\end{equation}
By definition, the limiting spectral density $\nu_S(\lambda)$ is obtained from the imaginary part of the resolvent:
\begin{equation}
\nu_S(\lambda) = - \frac{1}{\pi} \lim_{\epsilon \to 0^+} \text{Im}\, g_S(\lambda + i\epsilon).
\end{equation}

For $\lambda \in [\lambda_-, \lambda_+]$, the square root term in $g_{\tilde{S}}(z)$ contributes an imaginary part. A direct computation yields:
\begin{equation}
\lim_{\epsilon\to 0^{+}} \text{Im}\, g_{\tilde{S}}(\lambda + i\epsilon) \;=\; \frac{1}{2\rho \lambda} \sqrt{4\rho^2 \kappa \lambda - \left(\kappa \rho \lambda + \rho - \kappa \right)^2}.
\end{equation}
Substituting this into $g_S(z)$ gives:
\begin{equation}
\nu_S(\lambda) \;=\; \frac{\rho\kappa}{2\pi \lambda} \sqrt{(\lambda_+ - \lambda)(\lambda - \lambda_-)}, \quad \lambda \in [\lambda_-,\lambda_+],
\end{equation}
with the edges
\[
\lambda_\pm = \frac{1}{\rho\kappa} \left(\sqrt{\rho} \pm \sqrt{\kappa}\right)^2.
\]

Finally, since $S$ has rank at most $\min(r,M_0) = D\min(\rho,\kappa)$, there is a mass $1-\min(\rho,\kappa)$ at zero. 
Hence, we have:
\[
\nu_S(\lambda) \;=\; \left(1-\min(\rho,\kappa)\right)\delta(\lambda) \;+\; \mathbf{1}_{[\lambda_-,\lambda_+]}(\lambda) \, \frac{\rho\kappa}{2\pi \lambda} \sqrt{(\lambda_+ - \lambda)(\lambda - \lambda_-)}.
\]

\section{Consistency of the Replica Method with Numerical Experiments}
\label{app:agreement}
In this section, we validate our analytical results from the replica method against numerical experiments. As depicted in Figure~\ref{fig:check_order_parameters} and Figure~\ref{fig:check_generalization_error}, 
the theoretical predictions for both the order parameters and the generalization error show excellent agreement with the simulation results, confirming the accuracy of our results.

\begin{figure}[htb]
  \centering
  \includegraphics[width=1.0\textwidth]{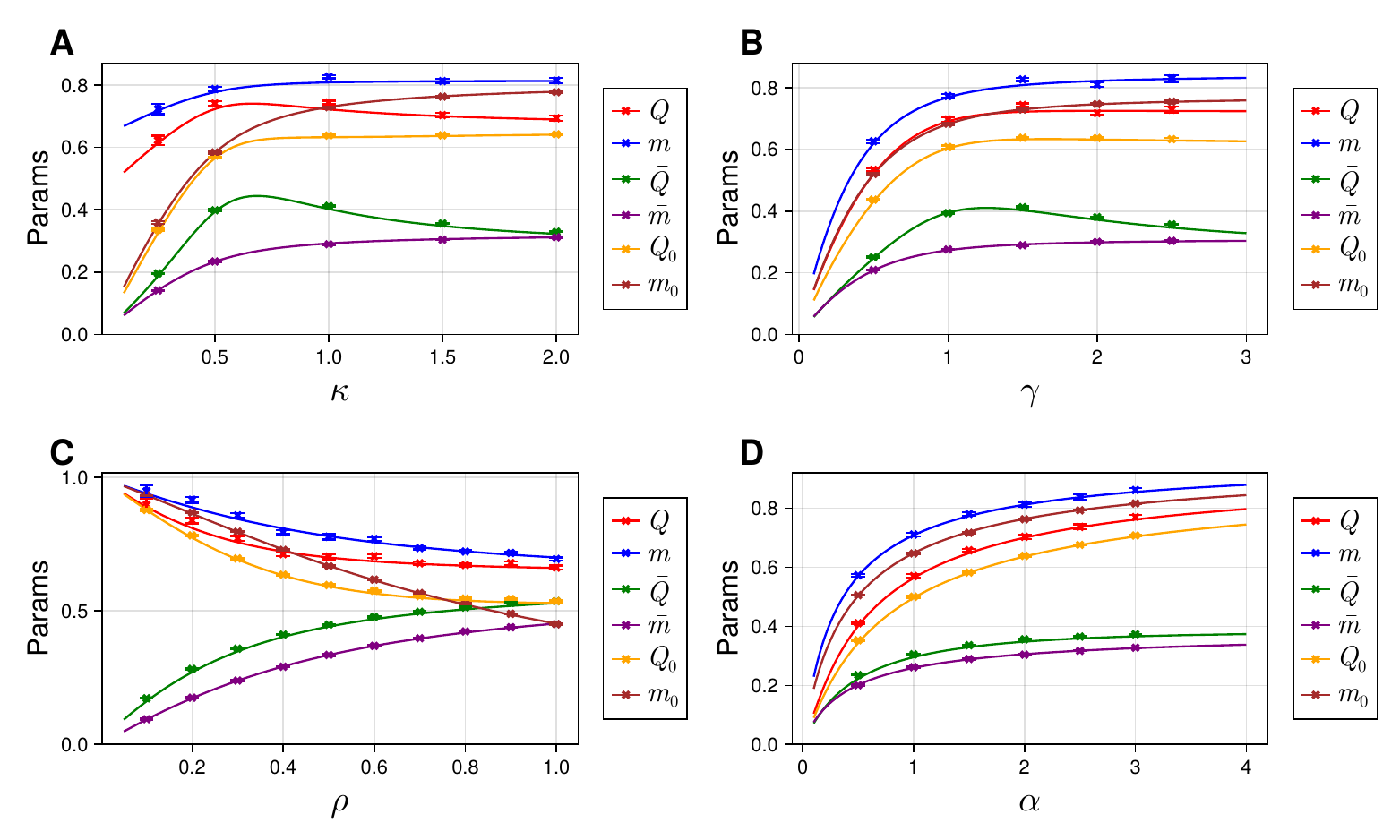}
  \caption{Comparison of the order parameters $Q, m, \bar{Q}, \bar{m}, Q_0, m_0$ obtained from the replica method (lines) with those obtained from the numerical experiments (error bars). 
  Parameters for (A-D): $\lambda=0.1, \sigma=0.1, D=70$; 
  (A)$\gamma=1.5, \rho=0.4, \alpha=2.0$; 
  (B)$\kappa=1.0, \rho=0.4, \alpha=2.0$; 
  (C)$\kappa=1.0, \gamma=1.5, \alpha=2.0$; 
  (D)$\kappa=1.0, \gamma=1.5, \rho=0.4$. Error bars represent the standard error of the mean over $10$ trials per point.}
  \label{fig:check_order_parameters}
\end{figure}

\begin{figure}[htb]
  \centering
  \includegraphics[width=1.0\textwidth]{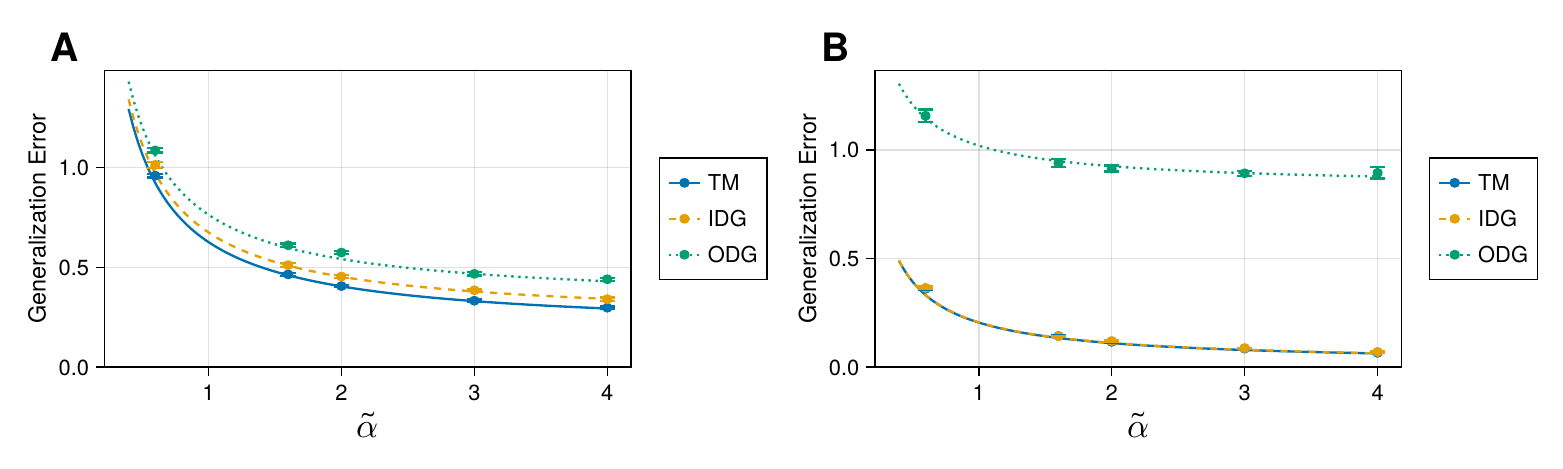}
  \caption{Comparison of the generalization errors ($\mathcal{E}_{\text{TM}}, \mathcal{E}_{\text{IDG}}, \mathcal{E}_{\text{ODG}}$) obtained from the replica method (lines) with those obtained from the numerical experiments (error bars). 
  Parameters for (A-D): $\kappa=4.0, \gamma = 0.5, \alpha=4.0, \lambda=0.1, \sigma=0.1, D=60$; 
  (A)$\rho=0.9$; 
  (B)$\rho=0.2$. Error bars represent the standard errors of the mean over $5$ trials per point.}
  \label{fig:check_generalization_error}
\end{figure}

\section{Experimental Details for the Softmax-Attention Experiments}
\label{app:experiments}

This section reports the reproducibility details of numerical experiments using a one-layer softmax self-attention model in Section~\ref{sec:experiments}. The task generation, prompt construction (including masking the query label), and the three evaluation protocols (TM/IDG/ODG) follow Section~\ref{sec:model}. We only describe the parts that differ from the linear-attention analysis.

\subsection{Architecture}
We replace the linear attention map in Eq.~\eqref{eq:decomp} with a standard scaled dot-product softmax attention layer with learned linear projections. Each token has dimension $D+1$, obtained by concatenating the input vector $\vb{x}\in\mathbb{R}^D$ and the scalar label channel. For a context matrix $C\in\mathbb{R}^{(D+1)\times (L+1)}$, we define
$Q = W_Q C,\quad K = W_K C,\quad V = W_V C$,
where $W_Q\in\mathbb{R}^{d_{\text{attn}}\times (D+1)}, W_K\in\mathbb{R}^{d_{\text{attn}}\times (D+1)}$, and $W_V\in\mathbb{R}^{(D+1)\times (D+1)}$. We use no biases in all projections. Unlike standard transformer blocks, we do not include an output projection $W_O$; instead, we set the value dimension to $D+1$ so that the value projection already returns the token dimension.

We use only the query vector of the last token. Writing 
\begin{align}
    q_{\text{last}}\in\mathbb{R}^{d_{\text{attn}}}
\end{align} 
for the last-column query and $K\in\mathbb{R}^{d_{\text{attn}}\times (L+1)}$ for all keys, the attention weights are
\begin{align}
    \vb{a} = \mathrm{softmax}\!\left(\frac{q_{\text{last}}^\top K}{\sqrt{d_{\text{attn}}}}\right)\in\mathbb{R}^{L+1}.
\end{align}
The scalar prediction is read out only from the label channel (the ($D+1$)-st coordinate) of the values:
$\hat{y} \;=\; \sum_{l=1}^{L+1} a_l\, V_{D+1,l}$.
This matches the evaluation convention in Section~\ref{sec:model}, where the query label is masked in the input and the model predicts the missing scalar.
Unless stated otherwise, we fix the attention width to $d_{\text{attn}}=4$ in all reported experiments.

\subsection{Training objective and optimization}
We train the model parameters $\Theta=\{W_Q,W_K,W_V\}$ by minimizing the empirical mean-squared error on the query label:
\begin{align}
    \mathcal{L}(\Theta)=\frac{1}{M}\sum_{\mu=1}^{M}\bigl(\hat{y}^\mu - y^\mu_{L+1}\bigr)^2,
\end{align}
using Adam with learning rate $\eta = 10^{-3}$. 
We use full-batch training: at each epoch, the loss $\mathcal{L}(\Theta)$ is evaluated over all $M$ training instances and a single Adam update is performed. 
All projection matrices are initialized using Glorot (Xavier) uniform initialization.

\vfill
%\bibliographystyle{apalike}
%\bibliography{ref}

\vfill
%\end{document}

\end{document}